\documentclass[onecolumn,a4paper,12pt,accepted=2020-02-26]{Quantumarticle}
\pdfoutput=1
\usepackage{amssymb,amsmath,graphicx,stmaryrd,mathtools, amsthm, color}
\usepackage{algorithm,algpseudocode}
\usepackage{enumitem}

\usepackage{authblk}
\usepackage{hyperref}

\newtheorem{theorem}{Theorem}

\newtheorem{proposition}[theorem]{Proposition}
\newtheorem{corollary}[theorem]{Corollary}
\newtheorem{definition}[theorem]{Definition}

\newcommand{\figcaption}[1]{\caption{\small #1}}

\theoremstyle{remark}

\newcommand{\tr}{{\operatorname{tr}}}
\newcommand{\ket}[1]{\left|#1\right\rangle}
\newcommand{\bra}[1]{\left\langle#1\right|}

\newcommand{\wsize}{{\rm{wsize}}}
\newcommand{\cT}{\mathcal T}
\newcommand{\cS}{\mathcal S}
\newcommand{\cG}{\mathcal G}
\newcommand{\cH}{\mathcal H}
\newcommand{\E}{\mathbb E}

\newcommand{\black}{{\rm{black}}}
\newcommand{\red}{{\rm{red}}}

\title{Quantum Speedup Based on Classical Decision Trees
}
\author[1]{Salman Beigi} 
\affil[1]{\it \small School of Mathematics, Institute for Research in Fundamental Sciences (IPM), Tehran, Iran}

\author[2]{Leila Taghavi}
\affil[2]{\it \small School of Computer Science, Institute for Research in Fundamental Sciences (IPM), Tehran, Iran}

\date{}

\begin{document}

\maketitle

\begin{abstract}
Lin and Lin~\cite{LL16} have recently shown how starting with a classical query algorithm (decision tree) for a function, we may find upper bounds on  its quantum query complexity. More precisely, they have shown that given a decision tree for a function $f:\{0,1\}^n\to[m]$ whose input can be accessed via queries to its bits, and a \emph{guessing algorithm} that predicts answers to the queries, there is a quantum query algorithm for $f$ which makes at most $O(\sqrt{GT})$ quantum queries where $T$ is the depth of the decision tree and $G$ is the maximum number of mistakes of the guessing algorithm. 
In this paper we give a simple proof of and generalize this result for functions $f:[\ell]^n \to [m]$ with non-binary input as well as output alphabets. 
Our main tool for this generalization is non-binary span program which has recently been developed for non-binary functions, and the dual adversary bound. As applications of our main result we present several quantum query upper bounds, some of which are new. In particular, we show that  topological sorting of vertices of a directed graph $\cG$ can be done with $O(n^{3/2})$ quantum queries in the adjacency matrix model. Also, we show that the quantum query complexity of the maximum bipartite matching is upper bounded by $O(n^{3/4}\sqrt {m + n})$ in the adjacency list model. 
\end{abstract}

\section{Introduction}
Query complexity of a function $f:[\ell]^n\to [m]$ is the minimum number of adaptive queries to its input bits required to compute the output of the function.  In a quantum query algorithm we allow to make queries in superposition, which sometimes improves the query complexity, e.g., in Grover's search algorithm~\cite{Grover96}.  

Lin and Lin~\cite{LL16} have recently shown that surprisingly sometimes classical query algorithms may result in improved quantum query algorithms. They showed that 
having a classical query algorithm with query complexity $T$ for some function $f:\{0,1\}^n\to[m]$, together with a \emph{guessing algorithm} that at each step predicts the value of the queried bit and makes no more than $G$ mistakes, the quantum query complexity of $f$ is at most $Q(f)=O(\sqrt{GT})$. For instance, the trivial classical algorithm for the search problem which queries the input bits one by one have query complexity $T=n$, and the guessing algorithm which always predicts the output $0$ makes at most $G=1$ mistakes (because making a mistake is equivalent to finding an input bit $1$ which solves the search problem). Thus the quantum query complexity of the search problem is $O(\sqrt{GT}) = O(\sqrt{n})$ recovering Grover's result.

There are two proofs of the above result in~\cite{LL16}. One of the proofs is based on the notion of  \emph{bomb query complexity} $B(f)$. Lin and Lin show that there exists a bomb query algorithm that computes $f$ using $O(GT)$ queries, and that the bomb query complexity equals the square of the quantum query complexity, i.e., $B(f)=\Theta(Q(f)^2)$, which together give $Q(f) = O(\sqrt{GT})$. In the second proof, they design a quantum query algorithm with query complexity $O(\sqrt{TG})$ for $f$ using Grover's search; in computing the function they use the values of predicted queries instead of the real values and use a modified version of Grover's search to find mistakes of the guessing algorithm. 

\paragraph{Our results:} In this paper we give a simple proof of the above result based on the method of  \emph{non-binary span program} that has recently been development by the authors~\cite{BT18}. Then inspired by this proof, we generalize Lin and Lin's result for functions $f:[\ell]^n\to [m]$ with non-binary input as well as non-binary output alphabets. Our proof of this generalization is based on the dual adversary bound which is another equivalent characterization of the quantum query complexity~\cite{LMRSS11}.

  
As an application of our main result we show that given query access to edges of a directed and acyclic graph $\cG$ in the \emph{adjacency matrix model}, 
the vertices of $\cG$ can be sorted with  $O(n^{3/2})$ quantum queries to its edges. We also show that given a directed graph $\cG$ and a vertex $v\in V(\cG)$, the quantum query complexity of determining the length of the smallest directed cycle in $\cG$ containing $v$ is $\Theta(n^{3/2})$. Moreover,  we show that given an undirected graph $\cG$, a vertex $v$ and some constant $k>0$, the quantum query complexity of deciding whether $\cG$ has a cycle of length $k$ containing the vertex $v$ is $O(n^{3/2})$. Furthermore, we show that some existing results on the quantum query complexity of graph theoretic problems such as directed $st$-connectivity, detecting bipartite graphs, finding strongly connected components, and deciding forests can easily be derived from our results. 

Our main result is also useful when dealing with graph problems in the \emph{adjacency list model}. In this regard,  we show that given query access to the adjacency list of an unweighted bipartite graph $\cG$, the quantum query complexity of finding a maximum bipartite matching in $\cG$ is $O(n^{3/4}\sqrt{m +n})$, where $m$ is the number of edges of the graph. To the authors' knowledge this is the first non-trivial upper bound for this problem. 

\section{Preliminaries}
In this section we review the notions of the dual adversary bound and the non-binary span program that will be used for the proof of our main result. 
In this paper we use Dirac's ket-bra notation, so $\ket v$ is a complex (column) vector whose conjugate transpose is denotes by $\bra{v}$. Then, $\langle v| w\rangle$ is the inner product of vectors $\ket v, \ket w$. 
For a matrix $A$, we denote  by $\|A\|$ the operator norm of $A$, i.e., the maximum singular value of $A$. We use
$[\ell]$  to denote the $\ell$-element set $\{0, \ldots, \ell-1\}.$
We also use the Kronecker delta symbol $\delta_{a,b}$ which equals $1$ if $a=b$ and equals $0$ otherwise.

\subsection{Query algorithms}\label{ssec:queryAlg}
In the query model we deal with the problem of computing a function $f:D_f\to[m]$ with domain $D_f\subseteq [\ell]^n$ by quering coordinates of the input $x= (x_1, \dots, x_n)\in D_f\subseteq  [\ell]^n$. 
In the classical setting a query algorithm asks the value of some coordinate of the input and based on the answer to that query decides what to do next: either asks another query or outputs the result. Such an algorithm can be modeled by a \emph{decision tree} whose internal vertices are associated with queries, i.e., indices $1\leq j\leq n$, and whose edges correspond to answers to queries, i.e., elements of $[\ell]$. At each vertex the algorithm queries the associated index, and then moves to the next vertex via the edge whose label equals the answer to that query. The algorithm ends once we reach the leaves of the tree that are labeled by elements of $[m]$, the output set of the function. The query complexity of the algorithm is the maximum number of queries in the algorithm over all $x\in D_f$, which is equal to the height of the decision tree.
A randomized classical query algorithm can similarly be modeled by a collection of decision trees where one of them is chosen at random.

In contrast in quantum query algorithms, a query can be made in superposition. Such a query to an input $x$ can be modeled by the unitary operator $O_x$:
\begin{equation*}
O_x|j,p\rangle=|j,(x_j+p) \mod \ell \rangle,
\end{equation*}
where the first register contains the query index $1\leq j\leq n$, and the second register saves the value of $x_j$ in a reversible manner. Therefore, a quantum query algorithm for computing $f(x)$ is an alternation of unitaries $O_x$ and some $U_i$'s that are independent of $x$ (but depend on $f$ itself). Indeed, a quantum query algorithm consists of sequence of unitaries
\begin{equation*}
U_kO_x\ldots U_2O_xU_1,
\end{equation*}
followed by a measurement which determines the outcome of the algorithm. We say that an algorithm computes $f$, if for every $x\in D_f\subseteq [\ell]^n$ the algorithm outputs  $f(x)$ with probability at least $2/3$.
The query complexity of such an algorithm is the number of queries, i.e., the number of $O_x$'s in the sequence of unitaries. $Q(f)$ denotes the \emph{quantum query complexity} of $f$, which is the minimum query complexity among all quantum algorithms that compute $f$.


\subsection{ Dual adversary bound}\label{ADVbound and dual}
The \emph{generalized adversary bound}~\cite{HLS07} gives a lower bound on the quantum query complexity of any function $f:D_f\to [m]$ with $D_f\subseteq [\ell]^n$. This bound can be obtained via a semi-definite program (SDP) whose optimal value, based on the duality of SDPs, has been shown to be equal to that of the following SDP up to a  factor of at most 2~\cite{LMRS10}.
\begin{subequations}\label{SDP:dual-SDP}
\begin{align} 
\min &\quad \max_{x\in D_f}\; \max\bigg\{\sum_{j=1}^n \big\| |u_{xj}\rangle \big\| ^2, \sum_{j=1}^n \big\| |w_{xj}\rangle \big\| ^2\bigg\}\\
\text{subject to} & \quad \sum_{j: x_j\neq y_j} \langle u_{xj}|w_{yj}\rangle=1-\delta_{f(x),f(y)} \qquad \forall x,y\in D_f.
\end{align}
\end{subequations}
Here the optimization is over vectors $|u_{xj}\rangle,|w_{xj}\rangle $. This SDP is called  
\emph{the dual adversary bound} and is proved by Lee et al.~\cite{LMRSS11} to be an upper bound on quantum query complexity of the function $f$ as well. Thus, the above SDP characterizes the quantum query complexity of $f$ up to a constant factor. Moreover, in order to design quantum query algorithms and quantum query complexity upper bounds, it is enough to find a feasible solution of the SDP~\eqref{SDP:dual-SDP}. 

Function evaluation is a special case of a more general problem called \textit{state generation}~\cite{Shi02, AMRR11}. In the state generation problem, the goal is to generate a state $\ket{\psi_x}$ (which depends on $x$) up to a constant error, given query access to $x\in D$. That is, the quantum query algorithm is required to output some state $\rho_x$ such that $\|\rho_x- \ket{\psi_x}\bra{\psi_x}\|_{\tr}\leq 0.1$ where $\|\cdot\|_{\tr}$ denotes the trace distance. 
Of course, the function evaluation problem is a special case of the state generation problem in which $\ket{\psi_x} = \ket{f(x)}$. It has been shown in~\cite{LMRSS11} that a generalization of the SDP~\eqref{SDP:dual-SDP} characterizes the quantum query complexity of the state generation problem up to a constant factor. This generalized SDP, again called the dual adversary bound, is as follows: 
\begin{subequations}\label{SDP:state-generation-SDP}
\begin{align}
\min &\quad \max_{x\in D}\; \max\bigg\{\sum_{j=1}^n \big\| |u_{xj}\rangle \big\| ^2, \sum_{j=1}^n \big\| |w_{xj}\rangle \big\| ^2\bigg\}\\
\text{subject to} & \quad \sum_{j: x_j\neq y_j} \langle u_{xj}|w_{yj}\rangle=1-\langle\psi_x\ket{\psi_y} \qquad \forall x,y\in D,
\end{align}
\end{subequations}
where again  the optimization is over vectors $|u_{xj}\rangle,|w_{xj}\rangle $. Observe that this SDP depends only on the gram matrix $\big(\langle\psi_x\ket{\psi_y} \big)_{x,y}$ of the target vectors. Moreover, letting $\ket{\psi_x} = \ket{f(x)}$, for some function $f$, we recover~\eqref{SDP:dual-SDP}.

\subsection{Non-binary span program}\label{sec:NBSP}

Span program introduced by~\cite{Rei09} is another algebraic tool that similar to the dual adversary bound, characterizes the quantum query complexity of binary functions up to a constant factor. This model has been used for designing quantum query algorithms of binary decision functions by \v{S}palek and Reichardt~\cite{RS12}. The notion of span program was generalized for functions with non-binary inputs in~\cite{ItoJeffery15}. Later, it was further generalized for arbitrary non-binary functions with non-binary input/output alphabets~\cite{BT18}.  In this paper we use a special form of non-binary span program of~\cite{BT18} called non-binary span program \emph{with orthogonal inputs}, which characterizes the quantum query complexity of any functions $f:[l]^n\to [m]$ up to a factor of $\sqrt{\ell-1}$. Here since we will use non-binary span programs only for functions with binary inputs ($\ell=2$), we may focus on this special form. 


A \emph{non-binary span program with orthogonal inputs} (NBSPwOI) $P$ evaluating a function $f:D_f\rightarrow[m]$ with $D_f\subseteq [\ell]^n$ consists of\footnote{Non-binary span programs may also have \emph{free input vectors} that will not be used here in this paper. }
\begin{itemize}
\item a finite-dimensional inner product space $V$,
\item $m$ target vectors $|t_0 \rangle, |t_2 \rangle,\ldots ,|t_{m-1} \rangle\in V$,
\item and an input set $I_{j,q}\subseteq V$ for every $1\leq j\leq n$ and $q\in [\ell]$.
\end{itemize}
Then $I\subseteq V$ is defined by 
$$I=\bigcup_{j=1}^n \bigcup_{q\in [\ell]} I_{j,q},$$
and for every $x\in D_f$ the set of \emph{available vectors} $I(x)$ is defined by 
$$I(x)=\bigcup_{j=1}^nI_{j,x_j}.$$
Indeed, when the $j$-th coordinate of  $x$ is equal to $q$ (i.e., $x_j=q$) then the vectors in $I_{j, q}$ become available.
We also let $A$ be the $d\times |I|$ matrix consisting of all input vectors as its columns where $d=\dim V$.

We say that $P$ evaluates the function $f$ if for every $x\in D_f$, $\ket{t_\alpha}$ belongs to the span of the available vectors $I(x)$ if and only if $\alpha=f(x)$. Even more, there should be two witnesses indicating this. Namely, a positive witness  $\ket{w_x} \in \mathbb C^{|I|}$ and a negative witness $\ket{\bar{w}_x}\in V$ satisfying the following conditions:
\begin{itemize}
\item The coordinates of $\ket{w_x}$ associated to unavailable vectors are zero. 
\item $A\ket{w_x} = \ket{t_\alpha}$. 
\item $\forall \ket v\in I(x)$ we have $\langle v|\bar w_x\rangle =0$. 
\item $\forall \beta\neq \alpha$ we have $\langle t_{\beta}|\bar{w}_x\rangle=1$. 
\end{itemize}


Let positive and negative complexities of $P$ together with the collections $w$ and $\bar w$ of positive and negative witnesses $(P,w,\bar{w})$ be
\begin{align*}
&\mathrm{wsize}^+(P,w,\bar{w}):=\max_{x\in D_f} ~\|\ket{w_x}\| ^2,\\
&\mathrm{wsize}^-(P,w,\bar{w}):=\max_{x\in D_f} ~\|A^\dagger \ket{\bar{w}_x}\| ^2.
\end{align*}
Then the complexity of $(P, w, \bar w)$ is equal to
\begin{align}\label{eq:wsize-pwwb}
 \mathrm{wsize}(P,w,\bar{w})=\sqrt{\mathrm{wsize}^-(P,w,\bar{w})\,\cdot\,\mathrm{wsize}^+(P,w,\bar{w})}. 
\end{align}

It is shown in~\cite{BT18} that for any NBSPwOI evaluating the function $f$, its complexity  $\mathrm{wsize}(P,w,\bar{w})$ is an upper bound on $Q(f)$. Furthermore, there always exists an associated NBSPwOI whose complexity is bounded by $O(\sqrt{\ell-1}Q(f))$. Thus, NBSPwOIs characterize the quantum query complexity of all functions up to a factor of $\sqrt{\ell-1}$.


\section{From decision trees to span programs}
In this section we first give a simple proof of the main result of~\cite{LL16} based on span programs. Later, getting intuition from this proof, we generalize this result for non-binary functions. 

Recall that a classical query algorithm for a function $f:D_f\to [m]$ with $D_f\subseteq \{0, 1\}^n$ can be modeled by a binary decision tree $\cT$ with internal vertices being indexed by elements of $\{1, \dots, n\}$, edges being indexed by $\{0,1\}$, and leaves being index by elements of $[m]$. 
The depth of the decision tree, which we denote by $T$, is the classical query complexity of this decision tree. In~\cite{LL16} it is assumed that there is a further algorithm that predicts the values of the queried bits. That is, at each internal vertex of $\cT$ it makes a guess for the answer of the associated query. This guess, of course, may depend on the answers to the previous queries.
Then it is proven that if for every $x\in D_f$ the number of mistakes of the guessing algorithm is at most $G$, then the quantum query complexity of $f$ is $O(\sqrt{TG})$.

We can visualize the guessing algorithm in the decision tree by coloring its edges. For each internal vertex of the decision tree, there are two outgoing edges indexed by $0$ and $1$, one of which is chosen by the guessing algorithm. We \emph{color} the chosen one black, and the other one red. We call such a coloring of the edges of the decision tree a \emph{guessing-coloring (hereafter, G-coloring)}. Now once we make a query at an internal vertex, its answer tells us which edge we should take, the black one or the red one. If it was black it means that the guessing algorithm made a correct guess, and if it was red it means that it made a mistake. 
Therefore, the number of mistakes of the guessing algorithm for every $x\in D_f$ equals the number of red edges in the path from the root to the leaf of the tree associated to $x$. 

Here we summarize the notion of G-coloring.
\begin{definition}[G-coloring]\label{def:G-coloring} A G-coloring of a decision tree $\cT$ is a coloring of its edges by two colors black and red, in such a way that any vertex of $\cT$ has at most one outgoing edge with black color. 
\end{definition}




We can now state the result of~\cite{LL16} based on decision trees and the notion of G-coloring.
 
\begin{theorem}[Lin and Lin~\cite{LL16}]\label{thm:binaryClassical2quantum}
Assume that we have a decision tree $\mathcal T$ for a function $f:D_f\to [m]$ with $D_f\subseteq \{0,1\}^n$ whose depth is $T$. Furthermore, assume that for a G-coloring of the edges of $\mathcal T$, the number of red edges in each path from the root to the leaves of $\mathcal T$ is at most $G$. Then there exists a quantum query algorithm computing the function $f$ with query complexity $O(\sqrt{GT})$.
\end{theorem}

We remark that the result of~\cite{LL16} also works for randomized algorithms. Nevertheless, here to present our main ideas we first consider deterministic decision trees. Later, randomized  query algorithms will be considered as well. 

To prove this theorem we design an NBSPwOI for $f$ with complexity $O(\sqrt{GT})$. To present this span program first we need to develop some notations. Let $V(\mathcal T)$ be the vertex set of $\mathcal T$. Then for every internal vertex $v\in V(\cT)$, its associated index is denoted by $J(v)$, i.e., $J(v)$ is the index $1\leq j\leq n$ that is queried by the classical algorithm at node $v$. The two outgoing edges of $v$ are indexed by elements of $\{0,1\}$ and connect $v$ to two other vertices. We denote these vertices by $N(v, 0)$ and $N(v, 1)$. That is, $N(v, q)$, for $q\in \{0,1\}$, is the next vertex that is reached from $v$ after following the outgoing edge with label $q$. We also represent the G-coloring of edges of $\cT$ by a function $C(v, q)\in \{\black, \red\}$ where $v$ is an internal vertex, $q\in \{0,1\}$ and $C(v, q)$ is the color of the outgoing edge of $v$ with label $q$.

\begin{proof}
For every $x\in D_f$ there is an associate leaf of the tree $\cT$ that is reached once we follow edges of the tree with labels $x_j$ starting from the root. In order to find $f(x)$ it suffices to find this associated leaf because this is what the classical query algorithm does; once we find the leaf associated to $x$, we find the path that the classical query algorithm would take and then find $f(x)$. Thus in order to compute $f$, we may compute another function $\tilde f$ which given $x$ outputs its associated leaf of $\mathcal T$, and to prove the upper bound of  $O(\sqrt{GT})$ on the quantum query complexity it suffices to design an NBSPwOI for $\tilde f$ with this complexity. 


The NBSPwOI is the following:
\begin{itemize}
\item the vector space $V$ is determined by the orthonormal basis indexed by vertices of $\mathcal T$:
$$\{\ket{v}\,|\, v\in V(\cT)\},$$ 
\item the input vectors are
$$I_{j,q}=\Big\{\sqrt{W_{C(v, q)}}\big(\ket{v}-\ket{{N(v, q)}}\big)  \,\Big|\, \forall v\in V(\cT) \text{ s.t. } J(v)=j \Big\},$$
where $W_{\black}$ and $W_{\red}$ are positive real numbers to be determined, 
\item the target vectors are indexed by leaves $u$ of the tree: 
$$\ket{t_u}=\ket{r}-\ket{u},$$
where $r\in V(\cT)$ is the root of the tree.
\end{itemize}

For every vertex $v$ of $\cT$ we denote by $P_v$ the (unique) path from the root $r$ to vertex $v$. 
Then for  every $x\in D_f$ there exists a path $P_x=P_{\tilde f(x)}$ from the root of the decision tree to the leaf $\tilde f(x)$. Thus the target vector $\big|t_{\tilde f(x)}\big\rangle$ equals
$$\big|t_{\tilde f(x)}\big\rangle= \ket r - \big|\tilde f(x)\big\rangle=\sum_{v\in P_{x}} \frac{1}{\sqrt{W_{C\big(v, x_{J(v)}\big)}}} \left\{ W_{C\big(v, x_{J(v)}\big)} \left(\ket{v}-\ket{N(v,x_{J(v)})}\right)\right\},$$
where the vectors in the braces are all available for $x$. Then since by assumptions the number of red edges along the path $P_{x}$ is at most $G$ and the number of all edges is at most $T$, the positive complexity is bounded by
$$\wsize^+\leq \frac{1}{W_{\red}}G+\frac{1}{W_{\black}}T.$$
We let the negative witness for $x$ to be 
$$\ket{\bar{w}_x}=\sum_{v \in P_{x} } \ket{v}.$$
It is easy to verify that $\ket{\bar w_x}$ is orthogonal to all available vectors, and that $\bra{\bar w_x} t_u\rangle = \bra{\bar w_x} r\rangle =1$ for all $u\neq \tilde f(x)$. Thus $\ket{\bar w_x}$ is a valid negative witness. Moreover, an input vector of the form
$$\sqrt{W_{C(v, q)}}\big(\ket{v}-\ket{{N(v, q)}}\big),$$
contributes in the negative witness size only if its corresponding edge $\{v, N(v, q)\}$ leaves the path $P_x$, i.e., they have only the vertex $v$ in common. In this case the contribution would be equal to $W_{C(v, q)}$, the weight of that edge. The number of such red (black) edges equals the number of black (red) edges in $P_x$, which is bounded by $ T$ ($G$). Therefore, the negative witness size is 
$$\wsize^-\leq  W_{\black} G+ W_{\red} T\big.$$
Now letting $W_{\black}=\frac{1}{W_{\red}}=\sqrt{\frac{T}{G}}$, both the positive and negative witnesses are bounded by $2\sqrt{GT}$. Therefore, the quantum query complexity of $\tilde f$, and then $f$ are bounded by $O(\sqrt{GT})$.
\end{proof}

\section{Main result: generalization to the non-binary case}

This section contains our main result which is a generalization of Theorem~\ref{thm:binaryClassical2quantum} for functions $f:D_f\to [m]$ with non-binary input alphabet $D_f\subseteq [\ell]^n$. In this case, a classical query algorithm corresponds to a decision tree whose internal vertices have out-degree $\ell$ (instead of 2). Moreover, a G-coloring can be defined similarly based on a guessing algorithm. Yet, we are interested in a further generalization of the notion of  decision tree which we explain by an example.


Consider the following trivial algorithm for finding the minimum of a list of numbers in $[\ell]$: we keep a candidate minimum, and as we query the numbers in the list one by one, we update it once we reach a  smaller number. In this algorithm, the possible numbers as answers to a query are of two types: numbers that are greater than or equal to the current candidate minimum, and those that are smaller. Now assuming that the answer to that query is of the first type, what we do next is independent of its exact value (since we simply ignore it and query the next index). Considering the associated decision tree $\cT$, for each vertex $v$ we have a candidate minimum, and the outgoing edges of $v$ are labeled by different numbers in $[\ell]$. Then by the above discussion, the subtrees of $\cT$ hanging below the outgoing edges whose labels are greater than or equal to the current candidate minimum are identical. Thus we can identify those edges and their associated subtrees. In this case the outgoing edges of $v$ are not labeled by elements of $[\ell]$, but by its certain subsets that form a partition. Indeed, there is an outgoing edge whose label is the \emph{subset} of numbers greater than or equal to the current candidate minimum, and an outgoing edge for any smaller number.

Motivated by the above example of minimum finding, we generalize the notion of decision tree $\cT$ for a function $f: D_f\to [m]$ with non-binary input alphabet ($D_f\subseteq [\ell]^n$). As before each internal vertex $v$ of $\cT$ corresponds to a query index $1\leq J(v)\leq n$. Each outgoing edge of this vertex is labeled by a subset of $[\ell]$, and we assume that these subsets form a partition of $[\ell]$.  We denote this partition by 
\begin{align}\label{eq:partition-Q}
\bigcup_{q=0}^{\ell-1} Q_v(q)=[\ell], 
\end{align}
where here $Q_v(q)$ is the subset in the partition that contains $q\in [\ell]$. Thus $Q_v(q)\subseteq [\ell]$ contains $q$, and for $q, q'\in [\ell]$ either $Q_v(q), Q_v(q')$ are disjoint or are equal. Moreover, the out-degree of $v$ equals $|\{Q_v(q)\,:\,  q\in [\ell]\}|$, the number of different $Q_v(q)$'s. We also denote the neighbor vertex of $v$ connected to the edge with label $Q_v(q)$ by $N(v, Q_v(q))$. See Figure~\ref{fig:2tree} for an example of a decision tree.

\begin{figure}[ht]
\centering \includegraphics[scale=.76]{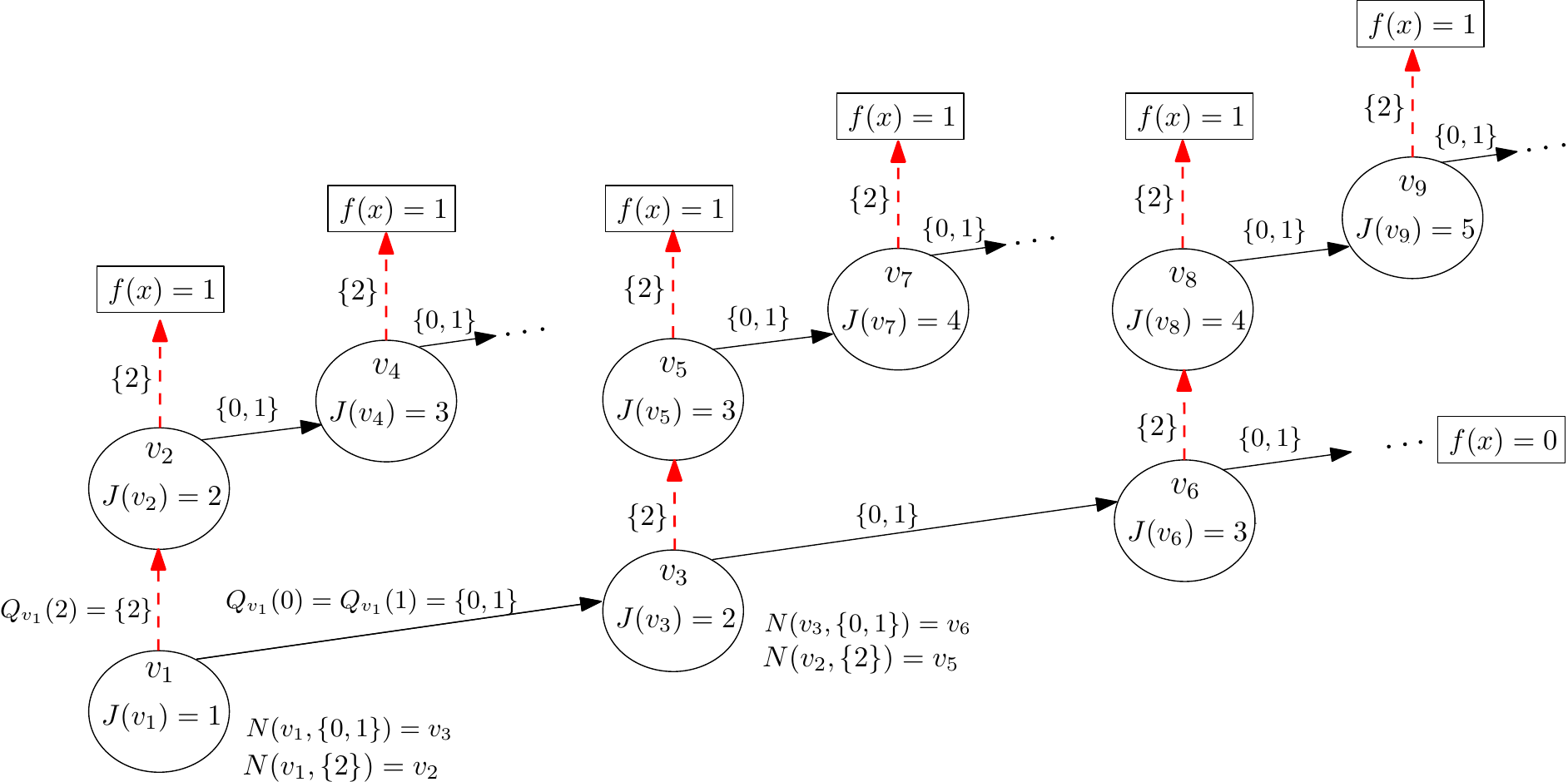}
\figcaption{Decision tree for deciding whether a given string $x\in \{0,1,2\}^n$ contains at least two 2's.
At any vertex $v$ the queried index is $J(v)$ and the result of the query belongs to one of the two sets appeared in the labels of outgoing edges of $v$.
This tree has a natural G-coloring: edges with label $\{2\}$ are red (dashed edges) and edges with label $\{0,1\}$ are black (solid edges).  The depth of the decision tree is $T=n$, and $f(x)$ would be determined once we see two red edges. Thus $G=2$ and the quantum query complexity of this problem is $O(\sqrt n)$. 
}
\label{fig:2tree}
\end{figure}

Now given a decision tree $\cT$ as above, the corresponding classical algorithm works as follows. We start with the root $r$ of the tree and query $J(r)$. Then $x_{J(r)}\in [\ell]$ corresponds to the outgoing edge of $v$ with label $Q_v(x_{J(r)})$. We take that edge and move to the next vertex $N(v, Q_v(x_{J(r)}))$. We continue until we reach a leaf of the tree which determines the value of $f(x)$.

The notation of G-coloring can also be generalized similarly. Recall that a G-coloring comes from a guessing algorithm that in each step predicts the answer to the queried index. In our generalized decision tree whose edges are labeled by subsets of $[\ell]$, we assume that the guessing algorithm chooses one of these subsets as its guess. Rephrasing this in terms of colors, we assume that for each internal vertex $v$ of $\cT$, one of its outgoing edges is colored in black (meaning that its label is the predicted answer) and its other outgoing edges are colored in red. We denote the color of the outgoing edge of vertex $v$ with label $Q_v(q)$ by $C(v, Q_v(q))\in \{\black, \red\}$.

Here is a summary of the notions of generalized decision tree and G-coloring explained above. 

\begin{definition}[Generalized decision tree and G-coloring]
A generalized decision tree $\cT$ is a rooted directed tree such that  each internal vertex $v$ (including the root) of $\cT$ corresponds to a query index $1\leq J(v)\leq n$. Outgoing edges of $v$ are labeled by subsets of $[\ell]$ that form a partition of $[\ell]$. We denote the subset that contains $q\in [\ell]$ by $Q_v(q)$ so that~\eqref{eq:partition-Q} holds. 
Leaves of $\cT$ are labeled with elements  of $[m]$. 

We say that $\cT$ decides a function $f: D_f\to [m]$ with $D_f\subseteq [\ell]^n$ if for every $x\in D_f$, by starting from the root of $\cT$ and following edges labeled by $Q_v(x_{J(v)})$ we reach a leaf with label $m=f(x)$.

As in Definition~\ref{def:G-coloring}, a G-coloring of a generalized decision tree $\cT$ is a coloring of its edges by two colors black and red, in such a way that any vertex of $\cT$ has at most one outgoing edge with black color.
\end{definition}\color{black}

We  also consider \emph{randomized classical query algorithms}. In this case, for each value $\zeta$ of the outcomes of some coin tosses, we have a (deterministic) generalized decision tree $\cT_\zeta$ as above. We also assume that each of these decision trees $\cT_\zeta$ is equipped with a guessing algorithm which itself may be randomized. Nevertheless, we may assume with no loss of generality that $\zeta$ includes the randomness of the guessing algorithm as well. Therefore, for any  $\zeta$ we have a generalized decision tree with a G-coloring as before. We assume that the classical randomized query algorithm outputs the correct answer $f(x)$ with high probability:
\begin{align}\label{eq:bounded-error}
\Pr_\zeta\big[\text{output of } \cT_\zeta \text{ on } x \text{ equals } f(x)\big]\geq 0.9.
\end{align}
The complexity of such a randomized query algorithm is given by the \emph{expectation} of the number of queries over the random choice of $\zeta$.

We can now state our generalization of Theorem~\ref{thm:binaryClassical2quantum}.

\begin{theorem}\label{thm:classical2quantum}
In the following let $f:D_f\to [m]$ be a function with $D_f\subseteq [\ell]^n$. 
\begin{itemize}
\item[{\rm(i)}] Let $\cT$ be a generalized decision tree for $f$ equipped with a G-coloring. Let $T$ be the depth of $\cT$ and let $G$ be the maximum number of red edges in any path from the root to leaves of $\cT$. Then the quantum query complexity of $f$ is upper bounded by $O(\sqrt{TG})$.

\item[\rm{(ii)}] Let $\{\cT_\zeta \,:\, \zeta\}$ be a set of generalized decision trees corresponding to a  \emph{randomized} classical query algorithm evaluating $f$ with bounded error as in~\eqref{eq:bounded-error}.
Moreover, suppose that each $\cT_\zeta$ is equipped with a G-coloring. Let $P_x^\zeta$ be the path from the root to the leaf of $\cT_\zeta$ associated to $x\in D_f$. Let $T_x^\zeta$ be the length of the path $P_x^\zeta$, and let $G_x^\zeta$ be the number of red edges in this path. Define
\begin{align*}
T=\max_x \mathbb{E}_\zeta[T_x^\zeta],\\
G=\max_x \mathbb{E}_\zeta[G_x^\zeta],
\end{align*}
where the expectation is over the random choice of $\zeta$. 
 Then the quantum query complexity of $f$ is  $O(\sqrt{TG})$. 
\end{itemize}
\end{theorem}

The span program in the proof of Theorem~\ref{thm:binaryClassical2quantum} can easily be adapted for a proof of the above theorem, yet in the complexity of the resulting span program we see an extra factor of $\sqrt{\ell-1}$, i.e., we get the upper bound of $O(\sqrt{(\ell-1)GT})$ on the quantum query complexity. To remove this undesirable factor, getting ideas from the span program in the proof of Theorem~\ref{thm:binaryClassical2quantum}, we directly construct a feasible solution of the dual adversary SDP~\eqref{SDP:dual-SDP}. Indeed, our starting point for proving Theorem~\ref{thm:classical2quantum} is the proof of Theorem~\ref{thm:binaryClassical2quantum} based on span programs. Then getting intuition from this proof, we design a feasible solution of the dual adversary SDP with the desired objective value.


\begin{proof} 
(i) 
Let $V_j(\cT)$ be the set of vertices of $\cT$ associated with query index $j$, i.e., $V_j(\cT) = J^{-1}(j)$. Also let $P_x$ be the path from the root $r$ to the leaf of $\cT$ associated to $x\in D_f$. We can assume with no loss of generality that $V_j(\cT)\cap P_x$ contains at most one vertex since otherwise in computing $f(x)$ we are querying index $j$ more than once.

To construct the feasible solution of the dual adversary SDP we will need the set of vectors $\{\ket{\mu_Q}:\, Q\subseteq [\ell] \}$ and $\{\ket{\nu_Q}:\, Q\subseteq [\ell]\}$ in $\mathbb C^{2^{[\ell]}}$ first appeared in~\cite{LMRS10}:
\begin{align}\label{eq:mu}
&\ket{\mu_Q}=\sqrt{\frac{2(2^\ell-1)}{2^\ell}}\left(-\theta\ket{Q}+\frac{\sqrt{1-\theta^2}}{\sqrt{2^\ell-1}}\sum_{P\neq Q}\ket{P}\right), \\ \label{eq:nu}
&\ket{\nu_Q}=\sqrt{\frac{2(2^\ell-1)}{2^\ell}}\left(\sqrt{1-\theta^2}\ket{Q}+\frac{\theta}{\sqrt{2^\ell-1}}\sum_{P\neq Q}\ket{P}\right),
\end{align}
where $\theta=\sqrt{\frac12-\frac{\sqrt{2^\ell-1}}{2^\ell}}$. These vectors have the property that $\|\ket{\mu_Q}\|^2=\|\ket{\nu_Q}\|^2=\frac{2(2^\ell-1)}{2^\ell}\leq 2$ for all $Q$ and 
$$\bra {\mu_Q} \nu_P \rangle=1-\delta_{Q, P}.$$
Also we use the set of vectors $\{\ket{\tilde\mu_\alpha}:\, \alpha\in [m] \}$ and $\{\ket{\tilde\nu_\alpha}:\, \alpha\in [m]\}$ in $\mathbb C^{m}$  defined similarly as above with the property that $\|\ket{\tilde\mu_\alpha}\|^2=\|\ket{\tilde \nu_\alpha}\|^2=\frac{2(m-1)}{m}\leq 2$ for all $\alpha$, and that
$\bra {\tilde \mu_\alpha} \tilde\nu_\beta \rangle=1-\delta_{\alpha, \beta}$.

Now define vectors $\ket{u_{xj}}$ and $\ket{w_{xj}}$ in the vector space $\mathbb C^{V(\cT)}\otimes \mathbb C^{\{\black, \red\}}\otimes \mathbb C^{2^{[\ell]}}\otimes \mathbb C^m$ as follows:
\begin{equation*}
\ket{u_{xj}}=\left\lbrace
\begin{array}{ll}
\frac{1}{\sqrt{W_{C(v,Q_v(x_j))}}}\big|{v , C(v,Q_v(x_j))}\big\rangle \otimes \ket{\mu_{Q_v(x_j)}} \otimes \ket{\tilde \mu_{f(x)}} ~& \text{if }\exists v\in P_{x}\cap V_j(\cT) \\ 
0 & {\rm otherwise,}
\end{array} 
\right. 
\end{equation*}
and
\begin{equation*}
\ket{w_{xj}}=\left\lbrace
\begin{array}{ll}
 \sum_{c\in C_{v, xj}}\sqrt{W_c}\ket{v, c} \otimes \ket{\nu_{Q_v(x_j)}}\otimes \ket{\tilde \nu_{f(x)}} ~& \text{if }\exists v\in P_{x}\cap V_j(\cT) \\ 
0 & {\rm otherwise,}
\end{array} 
\right.
\end{equation*}
where assuming that $v\in P_{x}\cap V_j(\cT)$,  $C_{v, xj}\subseteq\{\black, \red\}$ is defined by
\begin{align}\label{eq:def-C-xj}
C_{v, xj}=\big\{C(v,Q_v(q)): ~\, Q_v(q)\neq Q_v(x_j)\big\}.
\end{align}
\color{black}
Observe that assuming there is a (unique) vertex $v\in P_x\cap V_j(\cT)$, $\ket{u_{xj}}$ is defined in terms of the label and color of the outgoing edge of $v$ with label $Q_v(x_j)$. Moreover, $\ket{w_{xj}}$ is equal to either
$$\sqrt{W_{\red}} \ket{v, \red} \otimes \ket{\nu_{Q_v(x_j)}}\otimes  \ket{\tilde \nu_{f(x)}}, $$
or 
$$\Big(\sqrt{W_{\red}} \ket{v, \red}  + \sqrt{W_{\black}} \ket{v, \black}\Big) \otimes \ket{\nu_{Q_v(x_j)}}\otimes  \ket{\tilde \nu_{f(x)}},$$
depending on whether $C(v, Q_v(x_j)) = \black$ or $C(v, Q_v(x_j)) = \red$ respectively. 

We claim that these vectors form a solution of the SDP~\eqref{SDP:dual-SDP}. For every $x,y\in D_f$ with $ f(x)\neq  f(y)$ there exists a unique vertex $v\in V(\cT)$ such that $v\in P_{x}\cap P_{y}$ with  $Q_v^{x_{J(v)}}\neq Q_v^{y_{J(v)}}$ and in particular $x_{J(v)} \neq y_{J(v)}$. In this case,
$$\bra {u_{xJ(v)}} w_{yJ(v)}\rangle =1.$$
Moreover, for any $j\neq J(v)$, we have $\bra {u_{xj}} w_{yj}\rangle=0$ since for such $j$'s either one of $\ket{u_{xj}}, \ket{w_{yj}}$ is zero, or these vectors correspond to different vertices, or they correspond to the same vertex $v'\in P_x\cap P_y$ with $Q_{v'}(x_{J(v')}) = Q_{v'}(y_{J(v')})$ in which case $\ket{\mu_{Q_{v'}(x_{J(v')})}}$ and $\ket{\nu_{Q_{v'}(y_{J(v')})}}$ are orthogonal. Note that here we use the fact that if $f(x)\neq f(y)$ then $\bra{\tilde \mu_{f(x)}} \tilde \nu_{f(y)}\rangle =1$.
As a result,
$$\sum_{j:x_j\neq y_j}\bra{u_{xj}}w_{yj}\rangle=1. $$ 
Also if $f(x) =f(y)$ then since $\ket{\tilde \mu_{f(x)}}$ and $\ket{\tilde \nu_{f(y)}}$ are orthogonal we have 
$$\sum_{j:x_j\neq y_j}\bra{u_{xj}}w_{yj}\rangle=0. $$
Therefore, the vectors $\ket{u_{xj}}$ and $\ket{w_{xj}}$ form a feasible solution of the dual adversary SDP.

Now we compute the objective value. By assumption there are at most $ T$ edges in $P_x$ with black color, and at most $G$ red edges in $P_x$. Also the norm-squared of $\ket{\mu_Q}$'s and $\ket{\tilde \mu_\alpha}$'s are bounded by $2$. Therefore,
\begin{align*}
\sum_{j=1}^n \|\ket{u_{xj}}\|^2\leq 4\Big( \frac{1}{W_{\black}} T + \frac{1}{W_{\red}} G\Big).
\end{align*}
Also, in computing $\sum_{j=1}^n \|\ket{w_{xj}}\|^2$, for every vertex $v\in P_x$, if $C(v, Q_v(x_{J(v)}))=\black$ we get a term of $4W_{\red}$, and if $C(v, Q_v(x_{J(v)}))=\red$ we get a contribution of $4(W_{\black} + W_{\red})$. Now having a bound on the number of black and red edges in $P_x$ we find that
\begin{align*}
\sum_{j=1}^n \|\ket{w_{xj}}\|^2=4\Big( W_{\red}T +(W_{\black} +W_{\red})G \Big)\leq  4\Big( 2W_{\red}T +W_{\black}G \Big). 
\end{align*}
Therefore, if we let $W_{\black}=\frac{1}{W_{\red}}=\sqrt{\frac{T}{G}}$, then the objective value of the SDP~\eqref{SDP:dual-SDP} will be $O(\sqrt{GT})$.\\
\newline
\hspace{20pt}
\noindent
(ii) Let $f_\zeta: D_f\to [m]$ be the function that is computed by the decision tree $\cT_\zeta$. Then by assumption we have
\begin{align}\label{eq:f-zeta-exp}
\E_\zeta\big[\delta_{f_\zeta(x), f(x)}\big] \geq 0.9.
\end{align}
On the other hand, by part (i) for every $\zeta$ there is a feasible solution $\big|{u_{xj}^\zeta}\big\rangle$ and $\big|{w_{xj}^\zeta}\big\rangle$ of the dual adversary SDP for $f_\zeta$ with
\begin{equation*}
\sum_{j:x_j\neq y_j} \Big\langle u_{xj}^\zeta\ket{w_{yj}^\zeta}=1-\delta_{f_\zeta(x),f_\zeta(y)},
\end{equation*}
such that
\begin{align*}
\sum_{j=1}^n \big\|\big|{u^{\zeta}_{xj}}\big\rangle\big\|^2\leq 4\Big( \frac{1}{W_{\black}} T_x^\zeta + \frac{1}{W_{\red}} G_x^\zeta\Big),
\end{align*}
and 
\begin{align*}
\sum_{j=1}^n \big\|\big|{w^{\zeta}_{xj}}\big\rangle \big\|^2\leq  4\Big( 2W_{\red}T_x^\zeta +W_{\black}G_x^\zeta\Big). 
\end{align*}
Let us define
\begin{equation}\label{eq:u-zeta}
\ket{u_{xj}} = \frac{1}{\sqrt K}\sum_\zeta \big|{u_{xj}^\zeta}\big\rangle\otimes \ket \zeta,
\end{equation}
and
\begin{equation}\label{eq:w-zeta}
\ket{w_{xj}} = \frac{1}{\sqrt K}\sum_\zeta \big|{w_{xj}^\zeta}\big\rangle\otimes \ket \zeta,
\end{equation}
where $K$ is the number of possible values that $\zeta$ takes. 
Then we have
\begin{equation}\label{eq:f-zeta-x-y}
\sum_{j:x_j\neq y_j} \langle u_{xj}\ket{w_{yj}}=1-\frac1K\sum_\zeta \delta_{f_\zeta(x),f_\zeta(y)}.
\end{equation}
Now define
\begin{equation}\label{eq:psi-zeta}
\ket{\psi_x}:=\frac1{\sqrt{K}}\sum_\zeta\ket{f_\zeta(x)}\ket{\zeta},
\end{equation}
and consider the state generation problem for these vectors. Observe that
\begin{equation*}
\langle \psi_x\ket{\psi_y}=\frac1K \sum_\zeta \delta_{f_\zeta(x),f_\zeta(y)}.
\end{equation*}
Therefore, by~\eqref{eq:f-zeta-x-y} the vectors $\ket{u_{xj}}$ and $\ket{w_{xj}}$ form a feasible solution of the dual adversary SDP~\eqref{SDP:state-generation-SDP} for this state generation problem. Letting $M$ be the objective value of this SDP for these vectors, we conclude that with $O(M)$ quantum queries to $x$ we can generate a state $\rho_x$ such that $\|\rho_x - \ket{\psi_x}\bra{\psi_x}\|_{\tr}\leq 0.1$. Then measuring the first register of $\rho_x$ in the computational basis $\big\{\ket\alpha\,:\, \alpha\in [m], \big\}$ we have
\begin{align*}
\Pr[\text{measurement outcome equals }  f(x)]&= \tr\big[\rho_x \cdot \ket{f(x)}\bra{f(x)}\otimes I\big]\\
&\geq \tr\big[\ket{\psi_x}\bra{\psi_x} \cdot \ket{f(x)}\bra{f(x)}\otimes I\big] -0.1\\
& =\E_\zeta\big[ \delta_{f_\zeta(x), f(x)}\big]- 0.1\\
& \geq 0.9-0.1,
\end{align*} 
where in the last inequality we use~\eqref{eq:f-zeta-exp}. We conclude that there is a quantum query algorithm which makes $O(M)$ quantum queries and outputs $f(x)$ with probability at least $0.8$. Thus we only need to bound $M$, the objective value of the dual adversary bound. 

We compute \begin{align*}
\sum_{j=1}^n \|\ket{u_{xj}}\|^2 &=\frac{1}{K}\sum_\zeta \sum_{j=1}^n \big\|\big|{u^{\zeta}_{xj}}\big\rangle\big\|^2 \\
& \leq 4\frac{1}{K}\sum_\zeta \Big( \frac{1}{W_{\black}} T_x^\zeta + \frac{1}{W_{\red}} G_x^\zeta\Big)\\
& = 4\Big( \frac{1}{W_{\black}} \mathbb E_\zeta\big[T_x^\zeta\big] + \frac{1}{W_{\red}} \mathbb E_\zeta\big[G_x^\zeta\big]\Big)\\
&\leq 4\Big( \frac{1}{W_{\black}} T + \frac{1}{W_{\red}} G\Big)
\end{align*}
and similarly
\begin{align*}
\sum_{j=1}^n \|\ket{u_{xj}}\|^2\leq 4\Big( 2W_{\red} T +W_{\black}G\Big).
\end{align*}
Then as before letting $W_{\black}=\frac{1}{W_\red}=\sqrt{\frac{T}{G}}$, we find that the objective value of this feasible solution is bounded by $M=O(\sqrt{GT})$. We are done. 
\end{proof}


In the proof of Theorem~\ref{thm:classical2quantum} we assigned two different weights to edges of a decision tree based on their colors; the weight of any red edge is $W_\red$ and the weight of any black edge is $W_\black$. One may suggest that by assigning different wights to edges of $\cT$ we may get better bounds. That is, for any internal vertex $v$ of $\cT$, we may choose two weights $W_{v, \black}, W_{v, \red}$ and assign them to the outgoing edges of $v$ with the corresponding colors. Then the proof of Theorem~\ref{thm:classical2quantum} can be adopted to get a bound of the form $O(\max_{x, y} \sqrt{ M_x^+M_y^-})$ on the quantum query complexity where
\begin{align*}
&M_x^+=\sum_{\stackrel{v\in P_x:}{C(v, Q_v(x_{J(v)})) = \black}} \frac{1}{W_{v, \black}}+ \sum_{\stackrel{v\in P_x:}{C(v, Q_v(x_{J(v)})) = \red}} \frac{1}{W_{v, \red}}  ,\\
&M_x^-=\sum_{\stackrel{v\in P_x:}{C(v, Q_v(x_{J(v)})) = \black}} W_{v, \red}+ \sum_{\stackrel{v\in P_x:}{C(v, Q_v(x_{J(v)})) = \red}} W_{v, \black}.
\end{align*}
Then a simple application of the Cauchy-Schwartz inequality and $\max_{x, y} \sqrt{M_x^+ M_y^-}\geq \max_x\sqrt{M_x^+M_x^-}$ would show that 
updating the weights by
$$W'_{v, \black} = \frac{1}{W'_{v, \red}} = \sqrt{\frac{W_{v, \black}}{W_{v, \red}}},$$
would improve the upper bound $O(\max_{x, y} \sqrt{ M_x^+M_y^-})$. As a result, with no loss of generality we may assume that 
$$W_{v, \black} = \frac{1}{W_{v, \red}}.$$ 
Nevertheless, we still have the freedom to choose different weights for vertices of the decision tree $\cT$. These weights could depend on some parameter of the state of algorithm (decision tree) that is updated as we proceed. Moreover, it could depend on the guessing algorithm, e.g., on the number of red edges we have seen so far. In the following theorem, we analyze the latter option, and leave further investigation of this idea for future works. 



\begin{theorem}\label{thm:class2quantumTg}
 Let $\{\cT_\zeta \,:\, \zeta\}$ be a set of generalized decision trees corresponding to a  \emph{randomized} classical query algorithm evaluating $f$ with bounded error as in~\eqref{eq:bounded-error}.
Moreover, suppose that each $\cT_\zeta$ is equipped with a G-coloring. Let $P_x^\zeta$ be the path from the root to the leaf of $\cT_\zeta$ associated to $x\in D_f$. Let $G_x^\zeta$ be the number of red edges in $P_x^\zeta$, and for $1\leq g\leq G_x^\zeta$, let $T_{g,x}^\zeta$ be the number of black edges in $P_x^\zeta$ after the $g$-th red edge and before the next red one. Also let $T_{0, x}^\zeta$ be the number of black edges before the first red edge in $P_x^\zeta$, and let $T_{g, x}^\zeta =0$ for $g> G_x^\zeta$.
Let $G=\max_{x, \zeta} G_{x}^\zeta$ and define
\begin{align*}
T_{g}=\max_x \E_\zeta[T_{g,x}^\zeta],\qquad \quad 0\leq g\leq G.
\end{align*}
where the expectation is over the random choice of $\zeta$. 
 Then the quantum query complexity of $f$ is
\begin{equation*}
O \left(\sum_{g=1}^G \sqrt{{T}_{g}}\right).
\end{equation*}
\end{theorem}
\begin{proof} 
The proof is similar to the proof of Theorem~\ref{thm:classical2quantum} except that we pick different weights for edges of the decision trees. Using the notations we used before, for any choice of $\zeta$ and its associated decision tree $\cT_\zeta$ define
\begin{equation}\label{eq:u-T-i}
\ket{u_{xj}^\zeta}=\left\lbrace
\begin{array}{ll}
\frac{1}{\sqrt{W_{g(v),C(v,Q_v(x_j))}}}\big|{v , C(v,Q_v(x_j))}\big\rangle \otimes \ket{\mu_{Q_v(x_j)}} \otimes \ket{\tilde \mu_{f_\zeta(x)}} & \text{if } \exists v\in P^\zeta_{x}\cap V_j(\cT_\zeta) \\ 
0 & {\rm otherwise,}
\end{array} 
\right. 
\end{equation}
and
\begin{equation}\label{eq:w-T-i}
\ket{w_{xj}^\zeta}=\left\lbrace
\begin{array}{ll}
 \sum_{c\in C_{v, xj}}\sqrt{W_{g(v),c}}\ket{v, c} \otimes \ket{\nu_{Q_v(x_j)}}\otimes \ket{\tilde \nu_{f_\zeta(x)}} & \text{if }\exists v\in P_{x}^\zeta\cap V_j(\cT_\zeta) \\ 
0 & {\rm otherwise,}
\end{array} 
\right. 
\end{equation}
where as before $C_{v, xj}$ is given by~\eqref{eq:def-C-xj}, and $g(v)$ is the number of red edges in the path  from the root of $\cT_\zeta$ to $v$. Moreover, $W_{g, \black}, W_{g, \red}$, for any $g\geq 0$, are positive weights to be determined.  
As before, these vectors form a feasible solution of the SDP\eqref{SDP:dual-SDP} for the function $f_\zeta$. 
Then we define vectors $\ket{u_{xj}}$, $\ket{w_{xj}}$ and $\ket{\psi_x}$ as in~\eqref{eq:u-zeta}, \eqref{eq:w-zeta} and~\eqref{eq:psi-zeta}.
As before, we obtain a feasible solution to the SDP~\eqref{SDP:state-generation-SDP} whose objective value is an upper bound on the quantum query complexity of $f$. We estimate the objective value as follows. 

Let $W_g=W_{g,{\rm black}}=\frac{1}{W_{g,{\rm red}}},$ then
 \begin{align*}
\sum_{j=1}^n \|\ket{u_{xj}}\|^2 
 &=\frac{1}{K} \sum_\zeta \sum_{j=1}^n \big\|\big|{u^{\zeta}_{xj}}\big\rangle\big\|^2 \\&
\leq 4\frac{1}{K} \sum_\zeta \left(\frac1{W_{0}}T_{0,x}^\zeta +\sum_{g=1}^{G} \left(\frac 1{W_{g}}T_{g,x}^\zeta +W_{g}\right) \right)\\
&
=  4\left(\frac1{W_{0}}T_{0,x} +\sum_{g= 1}^G \left(\frac 1{W_{g}}T_{g,x} +W_{g} \right)\right)\\
&\leq  4\left(\frac 1{W_{0}} T_0+ \sum_{g= 1}^G \left(\frac1{W_{g}}T_g+W_{g}\right)\right).
\end{align*} 
Then letting $W_0=T_0$ and $W_g=\sqrt{T_g}$ for $g\geq 1$ we obtain\footnote{Note that $T_g\neq 0$ since every internal vertex of a decision tree has an outgoing black edge.}
\begin{equation*}
\sum_{j=1}^n \|\ket{u_{xj}}\|^2 
=O \left(\sum_{g= 1}^G \sqrt{T_{g}}\right).
\end{equation*}
We similarly obtain the same upper bound on $\sum_{j=1}^n \|\ket{w_{xj}}\|^2$. Then the quantum query complexity of $f$ is bounded by $O \left(\sum_{g= 1}^G \sqrt{T_{g}}\right)$.
\end{proof}

\section{Applications}

We can use our main result, Theorem~\ref{thm:classical2quantum}, to simplify the proof of some known quantum query complexity bounds as well as to derive new bounds. We start with some simple examples. 

\begin{proposition}\label{pro:example1} Suppose that we have query access to a list $x=(x_1,x_2,\ldots ,x_n)\in [\ell]^n$. Also let $q\in[\ell]$ and $1\leq k<n$ be fixed. 
\begin{enumerate}[label=(\roman*)]
\item[\rm{(i)}] \textsc{[counting]} The quantum query complexity of finding \emph{all} input indices with values equal to $q$ is $O(\sqrt{rn})$, where $\left|\big\{j:\,x_j=q\big\}\right| \leq r$.

\item[\rm{(ii)}] \textsc{[k-threshold]} The quantum query complexity of deciding whether $\left|\big\{ j:\,x_j=q\big\}\right|\leq k$ or not is $O(\sqrt{kn})$.
\end{enumerate}
\end{proposition}

It is shown that the quantum query complexity of 
 counting equals $\Theta(\sqrt{rn})$~\cite{BHT98}. Also it is well-known that the $k$-threshold problem has quantum query complexity $O(\sqrt{kn})$. 

\begin{proof}
\rm{(i)} In order to use Theorem~\ref{thm:classical2quantum} we first need a classical query algorithm. Suppose that we start from the first index and query all the indices one by one. We then output the set of indices $j$ with $x_j=q$. Next we need a G-coloring. To this end, observe that the algorithm is ignorant of the exact value of some index $x_j$ once it makes sure that $x_j\neq q$. Thus is the associated decision tree $\cT$ we can unify all outgoing edges of a vertex with label $q'\neq q$. That is, in $\cT$  there are two outgoing edges for any vertex that are labeled by $\{q\}$ and $[\ell]\setminus\{q\}$. Now we color all edges with label $\{q\}$ red and color the edges with label $[\ell]\setminus\{q\}$ black. In this coloring there are at most $r$ red edges in any path from the root to leaves: $G=r$. The depth of the decision tree is $T=n$. As a result the quantum query complexity of quantum counting is $O(\sqrt{rn})$.
\vspace{10pt}
\newline 
\rm{(ii)} The proof is similar to that of part (i). In the classical algorithm we query indices one by one until we find $k$ indices $j$ with $x_j=q$. Then in $\cT$ we unify edges with label $q'\neq q$ and color them black, and color edges with label $\{q\}$ red. As the algorithm stops once it faces $k$ indices with value $q$, the number of red edges in any path in $\cT$ from the root to leaves is at most $G=k$. Also the depth of the tree is $T=n$. Therefore the quantum query complexity of the threshold problem is  $O(\sqrt{kn})$. 
\end{proof}

\begin{proposition}\label{prop:min}
Let $x=(x_1, \dots, x_n)$ be a list of $n$ numbers. 
\begin{itemize}
\item[\rm{(i)}]\textsc{[min]} The quantum query complexity of finding $\min_j x_j$ is bounded by
$O(\sqrt{n \log n})$.

\item[\rm{(ii)}]\textsc{[k-min]} The problem of finding a subset $S\subseteq \{1, \dots, n\}$ of size $|S|=k$ such that for all $j\notin S$ we have $x_j\geq \max_{i\in S} x_i$ has quantum query complexity $O(\sqrt{kn\log n })$.

\end{itemize}
\end{proposition}

Two remarks are in line regarding the examples of minimum finding.  First, our bounds in these examples are tight only up to a factor of $\sqrt{\log n}$~\cite{Durr1996,DHHM04}. 
Yet, we would like to present these results since they show how randomization (part (ii) of Theorem~\ref{thm:classical2quantum}) may help to improve upper bounds on the quantum query complexity.

Second, observe that a list of numbers may have several minimums, so the problems in this proposition are not really function problems. To turn them into functions we may assume that our goal is to find the minimum number in the list whose index is also minimum. In other words, we consider a new order $``\prec"$ such that $x_i\prec x_j$ if $x_i<x_j$, or if $x_i=x_j$ and $i< j$. Now the minimum in this order is unique and we may ask for finding it. 

\begin{proof}
\rm{(i)} Consider the randomized classical algorithm that queries all indices one by one in a random order. The algorithm keeps a candidate for minimum at each step, and updates it once it reaches a smaller number.  Observe that this algorithm is ignorant of the exact answer to a query once it makes sure that it is not smaller than the current candidate for minimum.  Thus in the associated decision tree (for any choice of random order $\zeta$), at any internal vertex $v$ we can unify outgoing edges with label in $\{ q :\; q \geq m_v\}$ where $m_v$ is the candidate for minimum at node $v$. Thus in $\cT_\zeta$ any internal vertex $v$ has an outgoing edge with label $\{ q :\; q\geq m_v\}$ and an outgoing edge for any other  $q<m_v$. The former edge is colored black and the latter edges are colored red. The depth of $\cT_\zeta$ equals $T=n$ for any $\zeta$. However, for a given $x$, $G_x^\zeta$ depends on $\zeta$, so we should compute
$$G=\max_x \E_\zeta [G_x^\zeta].$$
We claim that $G=O(\log n)$. Intuitively speaking, the expected number of $x_j$'s that are smaller than the first queried element is $n/2$, and the guessing algorithm does not make mistakes once we query such $x_j$'s. Thus, after the first query, in expectation, half of the $x_j$'s would become irrelevant in computing $G$. Repeating this argument, we obtain $G=O(\log n)$. Below we present a more precise argument for this claim.

We can assume with no loss of generality that $x_1< \cdots < x_n$, since in the beginning of the algorithm we apply a random permutation. If in the random permutation $\zeta = (\zeta(1), \dots, \zeta(n))$ the first element is $n$, i.e., $\zeta(1)=n$, then $G_n^{\zeta} = G_{n-1}^{\zeta'} +1$ where $\zeta'=(\zeta(2), \dots, \zeta(n))$. Otherwise, if $\zeta(1)\neq n$ then $G_n^{\zeta} = G_{n-1}^{\zeta''}$ where $\zeta''$ is the same order as $\zeta$ from which $n$ is removed. We conclude that 
$$\E[G_n^\zeta] = \frac{1}{n} \big(\E\big[G_{n-1}^{\zeta'}\big]+1\big) + \frac{n-1}{n} \E\big[G_{n-1}^{\zeta''}\big].$$
Therefore, letting $G_n=\E[G_n^\zeta]$ we have
$$G_n = G_{n-1} + \frac 1 n.$$
Using $G_1=1$ we obtain 
$$G_n = \sum_{t=1}^n \frac 1 t= O(\log n).$$

As a result, $G=O(\log n)$ and by  Theorem~\ref{thm:classical2quantum} the quantum query complexity of finding the minimum is bounded by $O(\sqrt{n \log n})$.
\vspace{10pt}
\newline 
\rm{(ii)} The proof is similar to that of part (i). Again we read the numbers in a random order and update a $k$-list as our candidate for $S$ as we reach a number that is smaller than all the number in the list. The associated decision tree and its G-coloring is as before. Again we would have $T=n$. Also by similar ideas as in the proof of part (i) it can be shown that $G_n=G_{n-1}+k/n$ because with probability $k/n$ the largest $x_j$ appears in the first $k$ numbers in a random permutation. Therefore,
 $G=\max_x \E_\zeta[G^\zeta_x]=O(k\log n)$. 
We conclude that the quantum query complexity of finding the $k$ smallest numbers is bounded by $O(\sqrt{kn \log n})$.

\end{proof}
\color{black}

Motivated by Proposition~\ref{pro:example1} we can state the following general upper bound on the quantum query complexity of functions.

\begin{corollary}\label{cor:sparseProperty}
For any \emph{partial} function  $f:D_f\to[m]$ where $D_f\subseteq [\ell]^n$ and $\forall q\in [\ell]$, let 
$$r_{q}(x):=\big|\{j :\, x_j\neq q\}\big|  \qquad \text{ and } \qquad g= \min_{q\in [\ell]}\max_{x\in D_f} r_{q}(x).$$ 
Then if the classical query complexity of $f$ is $T$, the quantum query complexity of $f$ is $O(\sqrt{gT})$. In particular, the quantum query complexity of $f$ is $O(\sqrt{gn})$.
\end{corollary}
\begin{proof}
We prove this corollary using Theorem~\ref{thm:classical2quantum}. Given the classical algorithm for $f$, for a G-coloring of the edges of the associated decision tree, color every edge of the decision tree with label $q_0$ black and the rest of the edges red, where $q_0$ is such that $g=\max_{x\in D_f} r_{q_0}(x)$. Then since each $x\in D_f$ contains at most $g$ indices with values $q_0$, in every path from the root to leaves of the decision tree we see at most $G=g$ red edges. Then the quantum quantum query complexity of $f$ is $O(\sqrt{GT}) = O(\sqrt{gT})$.
\end{proof}

\subsection{Graph properties in the adjacency matrix model}
In this subsection and the following one we use Theorem~\ref{thm:classical2quantum} to prove quantum query complexity upper bounds on some graph theoretic problems. 
In this subsection, we assume that the graph is given in the \emph{adjacency matrix model}, by which we mean that the queries are from the entries of the adjacency matrix of the graph. That is, given vertices $u, v$ of the graph, we may ask whether there is an edge between $u$ and $v$ or not. Sometimes we assume that the underlying graph is directed in which case we ask whether there is a \emph{directed} edge from $u$ to $v$.

Inspired by the ideas in~\cite{LL16}, we make use of the well-known Breadth First Search algorithm (BFS, see Algorithm~\ref{alg:BFS}) as our starting point for designing classical algorithms for some graph theoretic problems.  The point of the BFS algorithm is that it returns a spanning tree (forest), with at most $n-1$ edges, of the underlying graph. Thus if we always guess that there is no edge between two queried vertices, we make at most $n-1$ mistakes.

\begin{algorithm}[ht]
\caption{ BFS$(\cG)$: breadth first search algorithm on graph $\cG$}
\label{alg:BFS}
\begin{algorithmic}[1]
\State Let $L$ be a list of unprocessed vertices and $Q$ be a first in first out queue. \label{algline-queue}
\State $L\leftarrow V(\cG)$, $Q=\emptyset$, $E_\cS=\emptyset$ \Comment $E_\cS$ stores the edge set of the BFS tree.
\While{there exists a $v'\in L$}
	\State add $v'$ to $Q$ 
	\State $L\leftarrow L- v'$ 
	\While{$Q \neq \emptyset$}
		\State $u\leftarrow {\rm dequeue}(Q)$
		\While{there exists a $v\in L$}
			\State Query $(u,v)$ 
			\If {$(u,v)\in E(\cG)$}
				\State add $(u,v)$ to $E_\cS$
				\State add $v$ to $Q$  \label{algline:BFS-add2Q}
				\State $L\leftarrow L- v$ 
			\EndIf
		\EndWhile
	\EndWhile
\EndWhile
\State \Return the BFS forest $\cS=\big(V(\cG),E_\cS\big)$
\end{algorithmic}
\end{algorithm}

\begin{proposition}
Suppose that we have query access to the adjacency matrix of a simple\footnote{We can derive the same results for non-simple graphs by making minor modifications in the proofs.} (possibly directed) graph $\cG$ on $n$ vertices. Then the followings hold.
\begin{enumerate}
\item[\rm{(i)}] \textsc{[bipartiteness]}\label{ex:bipartiteness} The quantum query complexity of deciding whether $\cG$ is bipartite or not is $O(n^{3/2})$.
\item[\rm{(ii)}] \textsc{[cycle detection]}\label{ex:cycle} The quantum query complexity of deciding whether $\cG$ is a forest or has a cycle is $O(n^{3/2}).$

\item[\rm{(iii)}] \textsc{[directed st-connectivity]}\label{ex:dir-st-con} The quantum query complexity of finding a shortest path (the path that consists of the least number of edges) between two vertices $s$ and $t$ in $\cG$ is $O(n^{3/2})$. This holds for either directed or undirected graphs.

\item[\rm{(iv)}]  \textsc{[smallest cycles containing a vertex]} 
The quantum query complexity of finding the length of the smallest \emph{directed} cycle containing a given vertex $v$ in a directed graph $\cG$ is $\Theta(n^{3/2})$. 

\item[\rm{(v)}]  \textsc{[$k$-cycle containing a vertex]}  The quantum query complexity of deciding whether $\cG$ has a cycle of length $k$, for a fixed $k$, containing a given vertex $v$ is $O((2k)^{(k-1)}n^{3/2})$.  

\end{enumerate}
\end{proposition}

The problem of bipartiteness has been first shown in~\cite{Ari15} to have quantum query complexity $O(n^{3/2})$, which is shown to be tight in~\cite{Zhang05}. 
An algorithm for the problem of cycle detection with $O(n^{3/2})$ queries is proposed in~\cite{CMB16} that works by reducing the problem to the \emph{st-connectivity} problem. This upper bound is known to be tight~\cite{ChK10}.
For the directed st-connectivity problem, it has been first shown to have query complexity $\Theta(n^{3/2})$ in~\cite{DHHM04}. There exists a quantum query algorithm for deciding whether $\cG$ contains a cycle of length less than $k$ containing a given vertex $v$ with query complexity $O(n\sqrt{k})$~\cite{CMB16}. For a list of related algorithms on cycle detection consult~\cite{Cira06}.

We would like to remark that the space complexity of all BFS/DFS-based quantum query algorithms in this subsection and the next one are linear in the size of the input graph. This is because our algorithms are based on feasible solutions of the dual adversary SDP that are obtained from a generalized decision tree. Now the point is that the 
space complexity of such an algorithm equals the \emph{logarithm} of the dimension of the vectors in the feasible solution of the dual adversary SDP, that itself equals the size of the decision tree which is exponential.

\begin{proof}
\rm{(i)} A graph $\cG$ is bipartite iff its vertices can be properly colored with two colors blue and green (such that no two adjacent vertices have the same color). Here is a classical algorithm to solve bipartiteness. We run the BFS algorithm (Algorithm~\ref{alg:BFS}) that outputs a spanning forest $\cS$ of $\cG$. Then we color every vertex of $\cG$ with odd depth in $\cS$  blue, and every vertex of $\cG$ with even depth in $\cS$  green. After this coloring, we search for an edge between two vertices with the same color in $\cG$. If no such edge exists, then $\cG$ is bipartite.
   
In order to use Theorem~\ref{thm:classical2quantum}, in the associated decision tree $\cT$ of the above algorithm, color every outgoing edge of $\cT$ with label $1$ red, and the rest of edges black. The depth of the decision tree is $T\leq n^2$ as the total number of possible queries (possible edges) for $\cG$ is $n(n-1)/2$. Also, by the above coloring of edges of $\cT$, we see at most  $n$ red edges in every path from the root to leaves of $\cT$. Indeed, we see at most $n-1$ red edges once we build the spanning forest $\cS$, and at most $1$ red edge once we search for an edge in $\cG$ between vertices with the same parity depths. Thus $G\leq n$ and the quantum query complexity of  bipartiteness is at most $O(\sqrt{GT}) =O(n^{3/2}).$ 
\vspace{10pt}
\newline 
\rm{(ii)} In a classical algorithm for this problem we first build a BFS forest and then search for an edge in the whole graph that does not belong to the BFS forest. If such an edge exists it should belong to a cycle in $\cG$. In order to use Theorem~\ref{thm:classical2quantum}, in the associated decision tree $\cT$, as before, we color every edge of $\cT$ with label $0$ black, and edges with label $1$ by red. The depth of the decision tree is $T\leq n^2$, and using this coloring in every path from the root to leaves of the decision tree there are at most $G= n$ red edges. Therefore, the quantum query complexity of the cycle detection problem is $O(n^{3/2}).$  
\color{black} 
\vspace{10pt}
\newline 
\rm{(iii)} Again we run the BFS algorithm on $\cG$ starting from vertex $s$ to build a subtree $\cS$ of $\cG$ with root $s$. Then a shortest path from $s$ to $t$, if exists,  belongs to $\cS$, and can be found once we have $\cS$.  The depth of the associated decision tree is $T=n^2$. For the G-coloring, as before, we color every edge with label 0 black and other edges red to get $G=n$. Then the  quantum query complexity of directed st-connectivity is $O(\sqrt{GT})=O(n^{3/2})$.
\vspace{10pt}
\newline 
\rm{(iv)}
In a classical algorithm for this problem we may run the BFS algorithm starting from vertex $v$. In parallel, whenever we reach a new vertex $u$ we query if there is an edge from $u$ to $v$. Finding such an edge corresponds to a smallest cycle containing $v$. 
As previous examples for a G-coloring of the associated decision tree, we color every edge with label 0 black and other edges red, then we have $G=n$ and $T=n^2.$ Therefore, the quantum query complexity of deciding whether $\cG$ has a cycle containing $v$ is $O(n^{3/2})$.

To prove the optimality of this bound we reduce the problem of directed st-connectivity which has query complexity $\Omega(n^{3/2})$ to this problem. Assume that we are given a graph $\cG$ and two distinguished vertices $s,t\in V(\cG)$, and we want to decide whether $s$ is connected to $t$ by a directed path or not. To solve this problem we build an auxiliary graph $\cH$ form $\cG$ as follows.
\begin{equation*}
V(\cH)=V(\cG)\cup \{w\}, \qquad
E(\cH)=E(\cG)\cup \{(w,s),(t,w)\}
\end{equation*}
Now $s$ is connected to $t$ in $\cG$ if and only if there is a directed cycle in $\cH$ containing the vertex $w$. Moreover, if such a cycle exists, its length equals the distance of $s, t$ in $\cG$ plus two. 
\vspace{10pt}
\newline 
\rm{(v)} In a classical algorithm for this problem, we first define an auxiliary directed graph $\cH$ out of $\cG$ with $V(\cG)=V(\cH)$.  To define the edge set of $\cH$ we use two random functions $C:V(\cG)\to [k]$ and $D:E(\cG)\to\{-1,+1\}$, and let
\begin{equation*}
E(\cH)=\big\{ (u,w)\in E(\cG)\,:\, C(u)=C(w)+1 (\mod k),\,  D(u,w)=+1\big\}.
\end{equation*}
Observe that if $\cG$ has a cycle of length $k$ containing $v$, then with probability at least $\frac{1}{(2k)^{k-1}}$, which is a constant, $\cH$ has a directed cycle of length $k$ containing $v$. Otherwise, $\cH$ does not have any cycle of length $k$ containing $v$.  
 Moreover, the length of all cycles of $\cH$ are multiples of $k$. Thus, the aforementioned cycle of $\cH$, if exists, is the smallest possible cycle. 
Then we can decide the existence of such a cycle using the algorithm of part~\rm{(iv)}. We can decide the existence of such a cycle with high probability by repeating the above algorithm $O\left((2k)^{(k-1)}\right)$ times.

\end{proof}

For the next set of examples we use the well-known classical algorithm Depth First Search (DFS). This algorithm builds a spanning forest of a given graph $\cG$. It is similar to the BFS algorithm but instead of using a queue which is a first in first out list, it uses a stack which is a first in last out list. This algorithm can also be implemented recursively (see Algorithm~\ref{alg:DFS}).

\begin{algorithm}[ht]
\caption{ DFS($\cG$): depth first search algorithm on graph $\cG$}
\label{alg:DFS}
\begin{algorithmic}[1]
\State let $L$ be a list of undiscovered vertices
\State let $ft$ be an array of size $|V(\cG)|$ \Comment{$ft$ stores the finishing time of vertices.}
\Function{DFS}{$\cG$}
\State $L\leftarrow V(\cG)$
\State $time=1$
	\While {there exists a $v\in L$}
		\State DFS$(\cG,v)$
	\EndWhile
	\State Return the DFS tree
\EndFunction

\Procedure{DFS}{$\cG,s$}
	\State $L\leftarrow L- s$
		\While{there exists a $v\in L$}
			\State Query $(s,v)$ 
			\If {$(s,v)\in E(\cG)$}
				\State DFS$(\cG,v)$
			\EndIf
		\EndWhile	
	\State $ft[s]\leftarrow time$
	\State $time\leftarrow time+1$
\EndProcedure
\end{algorithmic}
\end{algorithm}

\begin{proposition}\label{pro:DFS-based}
Suppose that we have query access to the adjacency matrix of a directed graph $\cG=(V,E)$ on $n$ vertices. Then the followings hold.
\begin{enumerate}
\item[\rm{(i)}] \textsc{[topological sort]}\label{ex:topsort} Suppose that $\cG$ is acyclic. Then the quantum query complexity of finding a vertex ordering of $\cG$ such that for all $(u,v)\in E$, $u$ appears before $v$ is $O(n^{3/2}).$

\item[\rm{(ii)}] \textsc{[connected components]}\label{ex:con-com} The quantum query complexity of determining connected components of $\cG$ is $O(n^{3/2})$. 

\item[\rm{(iii)}] \textsc{[strongly connected components]}\label{ex:str-con-com}
The quantum query complexity of finding strongly connected components of $\cG$ is $O(n^{3/2})$. Note that two vertices $u, v\in V$ belong to the same strongly connected component iff there exists a directed path from $u$ to $v$ and a directed path from $v$ to $u$ in $\cG$.
\end{enumerate}
\end{proposition}


The problem of topological sort is an important problem in large networks and job scheduling. There are several classical algorithms for this problem. The first algorithm is by Kahn~\cite{Kahn62}. In this algorithm at each step we add all vertices that do not have any incoming edges to the sorted list, and then eliminate them from the original graph. We continue this process until we add all vertices to the sorted list. Another algorithm for this problem, which we use in this proposition, is based on the DFS algorithm, first stated by Tarjan~\cite{Tarjan76}. Note that in these classical algorithms ones needs to read the entire input to discover the structure of the graph, so their query complexity is $O(n^2)$.
To the author's knowledge this proposition gives the first non-trivial quantum query complexity upper bound for the topological sort problem. The problem of finding (strongly) connected components of a (directed) graph has been first shown to have query complexity $\Theta(n^{3/2})$ in~\cite{DHHM04}.

\begin{proof}
\rm{(i)} For a classical algorithm for this problem,
run DFS and return vertices in their reverse of finishing time. For a G-coloring of the associated decision tree $\cT$, color every edge with label 0 black and every other edge red. Then as before there are at most $G= n$ red edges in every path from root to leaves of $\cT$. Also the depth of the decision tree is $T=n^2$. Thus we obtain the bound of $O(\sqrt{GT})=O(n^{3/2})$ on quantum query complexity of topological sort. 
\vspace{10pt}
\newline 
\rm{(ii)} We again use the DFS algorithm on $\cG$ and whenever the stack becomes empty a new connected component has been found. The G-coloring of the associated decision tree is as in part (i), and the bound of $O(n^{3/2})$ is derived similarly.
\vspace{10pt}
\newline 
\rm{(iii)} As a classical algorithm for this problem we use two DFS calls. In the first one we run the DFS algorithm on a reverse graph $\cG^R$ whose adjacency matrix is the transpose of the adjacency matrix of $\cG$, i.e., $(u,v)\in E(\cG^R)$ iff $(v,u)\in E(\cG)$. Observe that every query to $\cG^R$  is equivalent to a query to $\cG$. In the second one, the DFS will be run on the graph $\cG$ in the reverse finishing time ordering
\footnote{This is a reverse topological order of vertices of $\cG$. Therefore, a vertex at the end of this list is in a sink component.}
 of vertices from the first DFS run. Here we use the fact that if we start the DFS somewhere in a sink component\footnote{A sink component is a set of vertices $I\subseteq V(\cG)$ such that $\forall u\in I,v\in V(\cG)\setminus I$ we have $(u,v)\notin E(\cG)$.} then we exactly traverse that component. In the resulted DFS forest, vertices in every tree are in the same strongly connected component. 
For a G-coloring of the decision tree, we color every edge with label $0$ black and every other edges red, so that $G\leq 2n$. The depth of the decision tree is $T=n^2$. Therefore, the quantum query complexity of this problem is $O(n^{3/2})$. 
\end{proof}

The following corollary is a simple consequence of Corollary~\ref{cor:sparseProperty}.

\begin{corollary} The quantum query complexity of every graph property of a general graph\footnote{This applies to weighted, unweighted, directed or undirected graphs.} in the adjacency matrix model, is $O(n\sqrt{|E(\cG)|})$ which is faster than the trivial algorithm when $|E(\cG)|=o(n^2)$. In particular, every sparse graph property in the adjacency matrix model has quantum query complexity $O(n^{3/2})$.
\end{corollary}

The fact that any sparse graph property (particularly minor-closed graph properties) have quantum query complexity $O(n^{3/2})$ has been proven in~\cite{ChK10}.

\subsection{Graph properties in the adjacency list model}\label{sec:adjList}

In this subsection we present some bounds on the quantum query complexity of some graph properties when the underlying graph is given in the \emph{adjacency list model}. Let us first describe what we mean by this model.

In the adjacency list model we assume that the graph is given by an array of size $n(n-1)$ which for simplicity we think of it as a matrix of size $n\times (n-1)$. The $j$-th row of this matrix is a list of neighbors of the $j$-th vertex $v_j$ of the graph. Assume that $v_j$ has degree $d_{v_j}$. Then the first $d_{v_j}$ coordinates of the $j$-row contain the indices of the neighbors of $v_j$ (in some order), and the last $n-1-d_{v_j}$ coordinates are filled with a \emph{nil} symbol. See Figure~\ref{fig:BFS-graph-list} for an example. Any query in the adjacency list model corresponds to a pair $(v_j, i)$ with $i\leq n-1$. If $i\leq d_{v_j}$, then the output of this query is the $i$-th adjacent vertex of $v_j$ in $\cG$. If $i>d_{v_j}$, the output of this query is nil.
This model can also be defined for directed graphs similarly. The only difference is that the $j$-th row of the matrix contains vertices that can be reached from $v_j$ by a \emph{directed} edge.



\begin{figure}[ht]
\centering \includegraphics[scale=1]{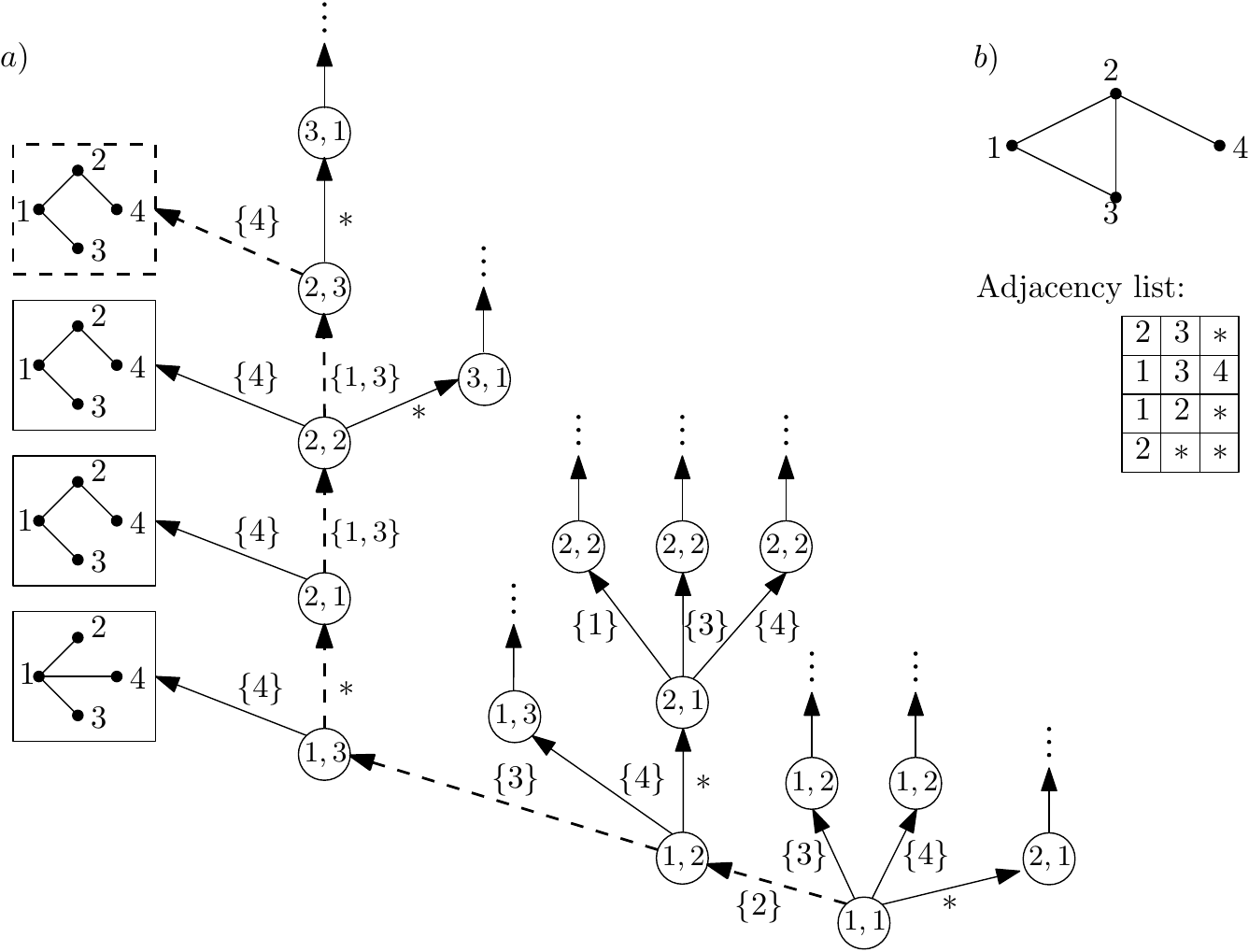}
\figcaption{$a)$  The decision tree for finding a BFS tree in a graph with 4 vertices. Each vertex of this decision tree is labeled by a pair which points to an entry of the matrix of adjacency list model. The nil symbol is represented by~$\ast$. The dashed path, is the path we take if the input to the BFS algorithm is the graph depicted on the right hand side.   $b)$ A graph and its adjacency list representation.}
\label{fig:BFS-graph-list}
\end{figure}

In the following we will use the BFS algorithm in the adjacency list model (see Algorithm~\ref{alg:BFSlist}) as a primitive to use Theorem~\ref{thm:classical2quantum}.  In the decision tree $\cT$ associated to this BFS algorithm, each node (query) corresponds to a pair $(v, i)$. The set of possible answers to such a query is the vertex set of $\cG$ which we partition as follows. We let $W(v, i)$ be the set of vertices that has been added to the BFS tree before querying $(v, i)$. The point is that the BFS algorithm is ignorant of the exact answer of the query $(v, i)$ once it makes sure that it belongs to $W(v, i)$ (see Figure~\ref{fig:BFS-graph-list} for an example). Thus in the decision tree $\cT$ we identify the outgoing edges of $(v, i)$ with labels in $W(v, i)$. All the other outgoing edges remain untouched. Now the G-coloring of $\cT$ is as follows: we color the outgoing edge of $(v, i)$ with label $W(v, i)$ black, and the rest of outgoing edges red. We note that there are $n$ vertices to be added to the BFS tree one-by-one, and we face a red edge once we add a new vertex or a nil. Then in total we see at most $G=O(n)$ red edges in every path from the root to leaves of $\cT$.
Also the total number of queries in the BFS algorithm equals the number of edges of $\cG$ denoted by $m=|E(\cG)|$ plus $n$. This is because as we do not know the degrees of vertices, we would stop querying neighbors of a vertex after seeing a nil symbol. This adds an extra query for every vertex. Thus $T=m+n$, and the quantum query complexity of finding the BFS tree in the adjacency list model is $O(\sqrt{GT}) =O(\sqrt{(m+n)n})$.

\begin{algorithm}[ht]
\caption{ BFS($\cG$): breadth first search algorithm on graph $\cG$ in adjacency list model}
\label{alg:BFSlist}
\begin{algorithmic}[1]
\State Let $W$ be a list of discovered vertices and $Q$ be a first in first out queue. \label{list-algline-queue}
\State $W\leftarrow \emptyset$, $Q\leftarrow \emptyset$, $E_\cS\leftarrow \emptyset$
\While{there exists a $v'\in V(\cG)\setminus L$}
	\State add $v'$ to $Q$ 
	\State $W\leftarrow W \cup \{v'\}$ 
	\While{$Q \neq \emptyset$}
		\State $u\leftarrow {\rm dequeue}(Q)$
		\State $v\leftarrow $Query $(u,1)$
		\State $i\leftarrow 2$
		\While{$v\neq$ nil}
			\State $v\leftarrow$ Query $(u,i)$ \Comment{returns the $i$-th neighbor of vertex $u$}
			\State $i\leftarrow i+1$
			\If {$v \in V(\cG)\setminus W$}
				\State add $(u,v)$ to $E_\cS$
				\State add $v$ to $Q$
				\State $P\leftarrow W \cup \{v\}$ 
			\EndIf
		\EndWhile
	\EndWhile
\EndWhile
\State \Return the BFS forest $\cS\big(V(\cG),E_\cS\big)$
\end{algorithmic}
\end{algorithm}

\begin{proposition}\label{pro:BFS-list-based} 
Suppose that the graph $\cG$ with $n$ vertices and $m$ edges is given via the adjacency list model. Then the following hold.
\begin{enumerate}
\item[\rm{(i)}] \textsc{[directed st-connectivity]}\label{ex:dir-st-con-list} Finding a shortest (directed or undirected) path between two vertices $s, t$ in $\cG$  has quantum query complexity $O\left(\sqrt{(m+n)n}\right)$.

\item[\rm{(ii)}] \textsc{[bipartitness]}\label{ex:bipartiteness-list} The quantum query complexity of deciding whether $\cG$ is bipartite or not is $O\left(\sqrt{(m+n)n}\right)$.

\item[\rm{(iii)}] \textsc{[maximum bipartite matching]} Assuming that $\cG$ is unweighted and bipartite, the quantum query complexity of finding a maximum bipartite matching in $\cG$ is $O\left(n^{3/4}\sqrt{m+n} \right).$ 

\item[\rm{(iv)}] \textsc{[topological sort]}\label{ex:topsort-list} Suppose that $\cG$ is acyclic. Then the quantum query complexity of finding a vertex ordering of $\cG$ such that for all $(u,v)\in E$, $u$ appears before $v$ is $O\left(\sqrt{(m+n)n}\right).$

\item[\rm{(v)}] \textsc{[connected components]}\label{ex:con-com-list} The quantum query complexity of determining connected components of $\cG$ is $O\left(\sqrt{(m+n)n}\right)$. 
 
\end{enumerate}
\end{proposition} 

Having query access to the adjacency list of a directed graph $\cG$, it has been proved in~\cite{DHHM04} that finding a minimum spanning tree of $\cG$ has quantum query complexity $O(\sqrt{mn})$. Using minimum spanning tree one can prove that checking directed st-connectivity and graph bipartiteness have quantum query complexity $O(\sqrt{mn})$ in the adjacency list model.
Lin and Lin~\cite{LL16} proved the upper bound of $O(n^{7/4})$ for the problem of maximum bipartite matching in the adjacency matrix model. Here using their ideas  we prove the first non-trivial upper bound for this problem in the adjacency list model.

\begin{proof}
\rm{(i)} To find a shortest path we run the BFS algorithm in the adjacency list model starting from the vertex $s$. Then $s$ and $t$ will be connected in the resulting spanning forest with their shortest path. As discussed before, the quantum query complexity of finding this BFS spanning forest is $O\left(\sqrt{(m+n)n}\right)$. Thus a shortest path between $s, t$ can be found with $O\left(\sqrt{(m+n)n}\right)$ quantum queries.
\vspace{10pt}
\newline 
\rm{(ii)} In the classical algorithm for this problem we start by finding a spanning tree on $\cG$ by running the BFS Algorithm~\ref{alg:BFSlist}. We then color vertices of $\cG$ using the resulting spanning forest $\cS$ with two colors blue and green. We color every vertex of $\cG$ with even depth in $\cS$ blue, and every vertex with odd depth in $\cS$ green. Then we search for two adjacent vertices in $\cG$ with the same color. If we find such an edge, the graph is not bipartite, and is bipartite otherwise. The G-coloring of the associated decision tree $\cT$ is as follows. In the first part that we run the BFS algorithm the G-coloring is as before. In the second part that we search for an edge between two vertices of the same color, we partition the set of possible answers (vertices of $\cG$) to in two parts: the set of blue vertices and the set of green vertices. As we query $(v, i)$, i.e., the $i$-th neighbor of $v$ in $\cG$, the color of the two outgoing edges associated to this query labeled by sets of blue and green vertices would be colored as follows: if $v$ is blue, the outgoing edge of blue vertices is colored red and the other one is colored black; if $v$ is green the outgoing edge of green vertices is colored red and the other one is colored black.  Observe that in the second part of the algorithm, once we see a red edge of $\cT$ the algorithm halts (and $\cG$ would not be bipartite). Thus in total we see at most $G=n$ red edges in any path from the root to leaves of $\cT$. 
On the other hand, the depth of the decision tree is $T=m+n$. Therefore, the quantum query complexity of this problem is $O\left(\sqrt{(m+n)n}\right)$.
\begin{algorithm}
\caption{Hopcroft-Karp algorithm for maximum bipartite matching on graph $\cG=(X\cup Y,E)$}
\label{alg:bipartite-matching}
\begin{algorithmic}[1]
\State $\mathcal{M}=\emptyset$ 
\Comment{$\mathcal{M}$ is an empty matching and will be updated until becoming a maximum matching} 
\While {$\mathcal{M}$ is not a maximum matching} \label{while-matching}
	\State define an auxiliary directed graph 			$\cG'=(V',E')$ as follows
	\begin{align*}
	E'=&\{(s,x)|x\in X,\forall y\in Y: (x,y)\notin 		\mathcal{M}\}\cup
	\{(y,t)|y\in Y, \forall x\in X: (x,y)\notin 			\mathcal{M}\} \\&
	\cup
	\{(x,y)|x\in X, y\in Y, (x,y)\notin \mathcal{M}\}
	\cup
	\{(y,x)|x\in X, y\in Y, (x,y)\in \mathcal{M}\}
	\\
	V'=& X\cup Y\cup\{s,t\}
	\end{align*}
	\Comment {any query to the adjacency list of $\cG'$ can be simulated using a query to the adjacency list of $\cG$}
	\State S= a maximal set of vertex disjoint shortest paths from $s$ to $t$ in $\cG'$
	 \Comment{ this can be found using one call to the algorithm~\ref{alg:BFSlist} in $\cG'$
}
	\If {$S=\emptyset$}
		\Return $\mathcal{M}$
	\Else
		\For {every path $(s,x_1,y_1,x_2,y_2,		\ldots,x_p,y_p,t)\in S$}
			\For {$i=1$ to $p-1$}
				\State $\mathcal{M}=\mathcal{M}-(x_{i+1},y_i)$
			\EndFor
			\For {$i=1$ to $p$}
				\State $\mathcal{M}=\mathcal{M}+(x_i,y_i)$
				\Comment{the size of $\mathcal{M}$ has been increased by 1}
			\EndFor
		\EndFor
	\EndIf
\EndWhile
\end{algorithmic}
\end{algorithm}
\vspace{10pt}
\newline 
\rm{(iii)} We use Algorithm~\ref{alg:bipartite-matching} by Hopcroft and Karp for maximum bipartite matching~\cite{HK70}.
In this algorithm we repeatedly increase the size of a partial matching $\mathcal{M}$ by finding augmenting paths in the graph. An augmenting path is a path with two end edges not in $\mathcal{M}$ and alternates between edges of the graph that belong to $\mathcal{M}$ and edges that do not. Swapping these edges from being in $\mathcal{M}$ to not being in $\mathcal M$ would increase the size of matching by one. However, instead of finding just an augmenting path in each iteration of the algorithm, it finds a maximal set of shortest vertex disjoint augmenting paths. After only $O(\sqrt{n})$ iterations, the maximum matching would be found.  
Since all queries to the input are made inside calls to the BFS Algorithm, the G-coloring of the associated decision tree, is as for BFS algorithm. There are $O(\sqrt{n})$ calls to BFS algorithm (Line~\ref{while-matching} in Algorithm~\ref{alg:bipartite-matching} repeats $O(\sqrt{n})$ times), so we have $G=n\sqrt{n}$ and the depth of the decision tree is $T=m+n$, where those $n$ extra queries are for the nils. Therefore, the quantum query complexity of this problem is $O\left(n^{3/4}\sqrt{(m+n)}\right)$. 
\vspace{10pt}
\newline 
\rm{(iv), (v)} The algorithms are similar to those of Proposition~\ref{pro:DFS-based} and the G-coloring is as above, so we skip the details. 
\end{proof}

%



%


\section{Concluding remarks}
In this paper we generalized a result of~\cite{LL16} that is a method for designing quantum query algorithms based on classical ones.  Our generalization of~\cite{LL16} is two-fold: first, we assume that the input alphabet of the function may be non-binary; second, we assume that in a decision tree the outgoing edges connected to a vertex may be indexed by subsets in a partition of the input alphabet set. These two enabled the possibility of using this method, in particular, for graph properties in the adjacency list model. Our proof of this generalization is based on span programs in the non-binary case as well as the dual adversary bound.

Let us at this stage review different approaches we have in proving Lin and Lin's results in~\cite{LL16} as well as Theorems~\ref{thm:classical2quantum} and~\ref{thm:class2quantumTg}:
\begin{itemize}
\item The first idea in~\cite{LL16} is to use the notion of bomb query complexity, which we did not mention here. It is an interesting question that whether this idea can be extended to prove our generalized results (Theorems~\ref{thm:classical2quantum} and~\ref{thm:class2quantumTg}).

\item The second idea in~\cite{LL16} is to use a modified version of Grover's search to find mistakes of the guessing algorithm. However, a naive application of Grover's search here results in an extra logarithmic factor for error reduction. 
It is shown in~\cite{LL16} that for functions with binary inputs this undesired factor can be eliminated using properties of the so called $\gamma_2$ norm. It seems plausible that the first part of Theorem~\ref{thm:classical2quantum} is provable by the same technique. However, it is not clear to us whether the second part of Theorem~\ref{thm:classical2quantum} or Theorem~\ref{thm:class2quantumTg} are achievable taking the same path.

\item The third idea is to use the notion of non-binary span program as we did for a proof of Theorem~\ref{thm:binaryClassical2quantum}. The idea is to use a ``st-connectivity type span program" (taken from~\cite{BR12}) in order to reach from the root of a decision tree to some leaf. However, to not end up with the trivial upper bound of $T$ (the depth of the decision tree) on the quantum query complexity, we equipped edges of the decision tree with some weights that are chosen based on a G-coloring.  Incorporating these weights in the span program the desired result was obtained. 

\item The last idea is to use the dual adversary bound. This approach is essentially the same as the approach of span programs, but with the advantage that it does not give an undesirable extra factor of $\sqrt{\ell-1}$ as explained before. Comparing to the first two methods, we believe that the ideas of using span programs and dual adversary bound are more advantageous since the choice of \emph{weights} in these approaches is arbitrary.   For proving  Theorem~\ref{thm:classical2quantum} the weights that we chose were among two possible choices. We then in Theorem~\ref{thm:class2quantumTg} showed how using a larger set of weights we may further improve the upper on the quantum query complexity.  Thus, a possible direction to extend our results is to further investigate other possible choices for the weights. 

\end{itemize}


One may suggest that our generalized non-binary version of the result of~\cite{LL16} can be obtained simply by representing non-binary inputs of $f:[\ell]^n\to [m]$ by binary strings, simulating a single non-binary query by $\log(\ell)$ binary ones, and then using the result of~\cite{LL16} in the binary case. Even ignoring the extra $\log(\ell)$ factor we obtain in this method, we argue that this approach does not work.  First, in our notion of generalized decision tree we allow to \emph{identify} some edges in the decision tree and label its edges with subsets of $[\ell]$. This is missing in the notion of decision tree in~\cite{LL16}. Identification of edges is a necessary part of our results especially in the examples of graph properties in the adjacency list model. To elaborate the second limitation of this approach, let us think of the example of minimum finding explained in Proposition~\ref{prop:min}. Suppose that $\ell=8$ and at some stage of the algorithm our candidate for minimum is 6 that is equal to $(1,1,0)$ in the binary representation. Then we read the first bit of the next number in the list and find it to be equal to 1. This means that the next number in the list is one of the numbers $4, 5, 6$ or $7$. In the algorithm and its associated G-coloring, there is a difference between 4, 5 and 6, 7 since the first two are smaller than 6. Indeed, in our proposed G-coloring edges 6,~7 are merged to a single edge with black color, while the edges 4 and 5 are colored in red. Therefore, to convert this coloring to a G-coloring in the binary decision tree whose edges are labeled by binary inputs, we have no choice but coloring the edge with label 1 by red. Then the parameter $G$ of the new $G$-coloring not only scales by a factor of $\log \ell$, but also is increased by something like $T-G$ because of such extra red edges.  
In summary, in order to use the result of~\cite{LL16} in the binary case to prove our generalized result in the non-binary case, we need to convert a G-coloring of a generalized non-binary decision tree to a G-coloring of a binary decision tree. It is now clear how this can be done in general without drastically weakening our bound on the parameter $G$.

Our results give bounds on the space complexity of our algorithms as well. The point is that the space complexity of a quantum algorithm based dual adversary bound, is bounded by the logarithm of the dimension of the vectors in the feasible solution of the dual adversary SDP. In our proofs the dimension of such feasible solutions is of the order of the size of the decision tree. Thus the space complexity of our algorithms equals the logarithm of the size of the corresponding decision tree. In particular, since in our examples (especially those for graph properties) the sizes of decision trees is exponential, the space complexity of our quantum algorithms is linear.\footnote{Note that although the size of the decision tree can be exponential, as in the examples in this paper, we do not need to explicitly build it. We usually have a classical algorithm which directly gives a decision tree. To use our results, we then only need to give a G-coloring.} 

 \color{black}
 
Prior to our work designing a span program based  quantum query algorithm for directed graphs was not an easy task. We eased the process of designing such algorithms by relating them to classical decision trees. Comparing to span programs for undirected graphs, however, the size of these span programs for directed graphs is exponential. 
 It would be interesting to see if we can decrease the space complexity of such quantum algorithms to logarithmic size.



\bibliography{references}
\bibliographystyle{alphaurl}


\end{document}